\documentclass[submission,copyright,creativecommons]{eptcs}
\providecommand{\event}{ICE '17}
 
\usepackage{breakurl}

\usepackage{hyperref}
\usepackage{amsfonts}
\usepackage{amsmath}
\usepackage{amssymb}
\usepackage{amsthm}
\usepackage{amsbsy}
\usepackage{enumerate}
\usepackage{stackrel}
\usepackage{bm}
\usepackage[toc,page]{appendix}
\usepackage{stmaryrd}
\usepackage{wasysym}
\usepackage{color}

\usepackage{derivation}

\usepackage{mathtools}
\usepackage{stmaryrd}

\usepackage{pifont}

\usepackage{cmll}

\newtheorem{definition}{Definition}[section]
\newtheorem{lemma}[definition]{Lemma}
\newtheorem{proposition}[definition]{Proposition}
\newtheorem{theorem}[definition]{Theorem}
\newtheorem{corollary}[definition]{Corollary}
\newtheorem{remark}[definition]{Remark}
\newtheorem{example}[definition]{Example}

\DeclareMathAlphabet{\mathpzc}{OT1}{pzc}{m}{it}

\newcommand{\Comment}[1]{ }

\def\Pred[#1]{~[\,#1\,]}

\newcommand{\cons}{\!:\!}

\newcommand{\congr}{\equiv}

\newcommand{\Orch}{\mbox{\sf Orch}}

\newcommand{\Synth}{\textbf{Synth}}

\newcommand{\SynthL}{\textbf{Synth}^{\!\textbf{fst}}}
\newcommand{\SynthA}{\textbf{Synth}^{\!\textbf{all}}}
\newcommand{\SynthUD}{\textbf{Synth}^{\!\!\textbf{\tiny UD}}}

\newcommand{\io}{\bullet}

\newcommand{\dom}{\textit{dom}}

\newcommand{\oftype}{\rhd}

\newcommand{\bvdash}{\bm{\vdash}\hspace{-5.5pt}\bm{\vdash}}

\newcommand{\List}[1]{[ #1 ]}

\newcommand{\namedorch}[2]{\langle #1 \rangle #2}

\newcommand{\outputtype}[1]{\bm{!}[#1]}
\newcommand{\inputtype}[1]{\bm{?}[#1]}
\newcommand{\branchtype}[1]{ {\bm \with} \hspace{-2pt}\Set{ #1}}

\newcommand{\selectiontype}[1]{\mbox{\large ${\bm \oplus}$}\hspace{-2pt} \Set{ #1 }}

\newcommand{\retrseltype}[1]{{\bm \boxplus} \Set{ #1 }}
\newcommand{\retrselop}[2]{#1 \! \triangleleft [ #2 ]}
\newcommand{\prtyretrseltype}[1]{{\bm \boxplus} \!\, \langle\!\langle #1 \rangle\!\rangle}
\newcommand{\prtyretrselop}[2]{#1 \! \triangleleft \langle\!\langle  #2 \rangle\!\rangle}
\newcommand{\Labels}{{\cal L}}
\newcommand{\orchF}{\mathsf{f}}
\newcommand{\orchG}{\mathsf{g}}
\newcommand{\orchPlus}{\mbox{\small $\sum$}}
\newcommand{\orchOplus}{\oplus}
\newcommand{\orchComply}[3]{#1 : \, #2 \comply #3}

\newcommand{\inact}{\mbox{$\mathbf{0}$}}

\newcommand{\typedaccept}[4]{\mathtt{accept}_{#2} (#3)#4}
\newcommand{\typedrequest}[4]{\mathtt{request}_{#2} (#3)#4}

\newcommand{\send}[2]{#1 ! [#2]}
\newcommand{\receive}[2]{#1 ? (#2)}
\newcommand{\catch}[3]{\mbox{\tt catch\,} #1 (#2).#3}

\newcommand{\throw}[3]{\mbox{\tt throw\,} #1 [#2].#3}

\newcommand{\select}[2]{#1 {\triangleleft\,} #2}
\newcommand{\branch}[2]{#1 \triangleright #2}

\newcommand{\prtysel}[2]{#1\, _{\mathtt{p\hspace{-1pt}r\hspace{-1pt}t\hspace{-1pt}y}}\!\!\triangleleft #2}

\newcommand{\Subst}[2]{\{#1/#2\}}

\newcommand{\IfThenElse}[3]{\mbox{\tt if}\; #1 \; \mbox{\tt then}\; #2 \; \mbox{\tt else} \; #3}

\newcommand{\fc}[1]{\mbox{\sc fc}(#1)}

\newcommand{\comply}{\dashv}

\newcommand{\Dual}[1]{\overline{#1}}

\newcommand{\xs}{\textsf{xs}}

\newcommand{\Set}[1]{\{#1\}}

\newcommand{\lts}[1]{\stackrel{#1}{\longrightarrow}}

\newcommand{\notlts}[1]{\stackrel{#1}{\;\;\not\!\!\longrightarrow}}

\newcommand{\stopf}{\mathfrak{1}}
\newcommand{\stopA}{\mathsf{end}}

\newcommand{\Actdot}{.}

\newcommand{\dquad}{\quad\quad}

\newcommand{\IF}{\textbf{if}}
\newcommand{\THEN}{\textbf{then}}
\newcommand{\WHERE}{\textbf{where}}
\newcommand{\ELSE}{\textbf{else}}
\newcommand{\AND}{\textbf{and}}
\newcommand{\OR}{\textbf{or}}
\newcommand{\LET}{\textbf{let}}

\newcommand{\FAIL}{\textbf{fail}}
\newcommand{\IN}{\textbf{in}}

\newcommand{\bigfract}[2]{\frac{^{\textstyle #1}}{_{\textstyle #2}}}
\newcommand{\ruleName}[1]{[{\sc #1}]}
\newcommand{\ruleSpace}{\vspace{0em}}
\def \mathax #1#2{\ruleSpace\begin{array}{l}{\mbox{\scriptsize \ruleName{#1}} } \\[-.1cm] #2
\end{array}}
\def \mathrule #1#2#3{\ruleSpace\begin{array}{l}%
    {\mbox{\scriptsize \ruleName{#1}}}
    \\ \bigfract{#2}{#3}
\end{array}}

\def\titlerunning{Session Types for Orchestrated Interactions}
\def\authorrunning{
Barbanera, de'Liguoro
}

\title{Session Types for Orchestrated Interactions\footnote{This work was partially 
supported by the COST Action IC1405 on ``Reversible computation - extending horizons of computing'' and
by the COST Action EUTYPES CA-15123.
The authors were partially supported also
by, respectively,
Project FIR 1B8C1 of the University of Catania and
Project FORMS 2015 of the University of Torino.
}}

\author{Franco Barbanera
\institute{Dipartimento di Matematica e Informatica\\
University of Catania}
\email{barba@dmi.unict.it}
 \and 
Ugo de'Liguoro
\institute{Dipartimento di Informatica\\
University of Torino}
\email{ugo.deliguoro@unito.it}
}

\begin{document}

\setlength{\abovedisplayskip}{6pt}
\setlength{\belowdisplayskip}{\abovedisplayskip}

\maketitle

\begin{abstract}
In the setting of the $\pi$-calculus with binary sessions, we aim at relaxing the notion of duality of session types by 
the concept of {\em retractable compliance} developed in contract theory. 
This leads to extending session types with a new type operator of
``speculative selection" including choices not necessarily
offered by a compliant partner. We address the problem of  selecting successful communicating branches 
by means of an operational semantics based on orchestrators, which has been shown to be equivalent to 
the retractable semantics of contracts, but clearly more feasible.
A type system, sound with respect to such a semantics, is hence provided.
\end{abstract}


\section{Introduction}
\label{sect:Intro}

Contracts \cite{LP07,LP08,CGP10,BH16} and session types \cite{HondaVK98,YoshidaV07} are both intended as abstractions representing interaction protocols among
concurrent processes. In both theories, interaction is modeled by message exchange along channels, abstractly represented by input/output actions indexed over channel names. Also the used formalisms all stem from CCS and its variants, possibly extended by some value passing mechanism. The resemblance is even tighter when considering ``session contracts'' \cite{BdL10,BH16}, where only internal and external choices among contracts prefixed by pairwise distinct output viz. input actions are respectively allowed.

{\bf Contracts versus Session Types.}
In spite of similarities, these theories stem from rather different concepts. 
In case of contracts, the input/output behaviour of a participant to a conversation is formalised as whole by a term of an appropriate process algebra; contract theory then focuses on the ``compliance'' relation, holding when two or more protocols are such that, whenever there is an action by a participant that is expected to be performed, the symmetric namely dual one is made available by some other participant. Restricting to the binary case, we say that a ``server'' protocol is compliant with a ``client'' one if all actions issued by the latter are matched by the respective co-actions by the former, possibly until the client reaches a successful state.

Session types are a type system for a dialect of Milner's $\pi$-calculus. Like with typed $\pi$-calculus, judgments associate to a process a ``typing'' that pairs channel names with the types of the values that can be transmitted through the channels; since $\pi$-terms are allowed to communicate channel names as well, channel types are among the types of exchanged values. Differently from ordinary $\pi$-calculus types, session types are regular trees of value types, that can be session types as well. In this way a single type can describe the flow of data through each channel that do not need to have all the same type; also input/output communication actions are distinguished by their types. When a ``session'' is opened by two processes in parallel, a new private channel is created - the session channel - that is shared by the processes; if the two ``end-points'', namely the respective occurrences of the session channel in the participant processes, are typed by {\em dual} types - roughly interchanging input and output types - then the interaction will be error free at run-time. Observe that not only the typing has to be checked against the process structure, which is not considered in contract theory, but also the very same process can issue several sessions at the same time, and session channels can be exchanged among processes. Therefore, even in the simpler setting of binary sessions, namely with channel names connecting two processes at a time, difficulties arise from the possible nesting of several sessions, and their ability to communicate across the boundary of a single session. 

{\bf Orchestrated Compliance.}
Compliance being a rather restrictive requirement, more liberal constraints have been proposed in the literature, among which are ``orchestrated'' compliance \cite{Padovani10} and ``retractable'' contracts \cite{BDLdL16}. 
According to the orchestrated model the interactions between a client and a server are mediated by a third process - the {\em orchestrator} - ruling interactions by allowing certain actions and co-actions only, possibly buffering a bounded number of messages on both sides. In the retractable model instead, actions are classified into irretractable (unaffectible) and retractable (affectible) ones. The concept is that, while irretractable actions by a participant have to be matched with their duals by the other compliant participant, retractable actions are just tried and possibly retracted, in case of communication failure, to issue some other action instead. Although these two models are different, it has been shown in \cite{BdL16} that, by restricting to certain orchestrators that allow just synchronous  communications, contracts that are deemed compliant in both models are the same. Moreover, it is possible to provide an algorithm which synthesizes an orchestrator out of two retractable contracts if they are compliant, or fails otherwise (see \cite{BvBdL17} and section \ref{sec:priorities} below).

{\bf Orchestrated interaction instead of duality.}
In this paper we address the issue of adapting the idea of retractable contracts to session types. More precisely we see session types as contracts, and propose to replace the otherwise restrictive notion of type duality by the relation of retractable or equivalently orchestrated compliance. 
To better illustrate the point, let us  consider a process $P$ 
in parallel with a system $Q$, that has to choose
 how to interact in a session with $Q$ by selecting one of several alternatives. 
Both $P$ and $Q$ are equipped with specification of their behaviours, so that it is known in advance
that at least one alternative is actually successful, but not necessarily all of them are.
Now there are two possible ways of guaranteeing $P$ 
to successfully complete the interaction with  $Q$:
\begin{itemize}
\item either $P \mid Q$ 
	is run on a computational infrastructure
	that, in case of  a synchronization failure, allows to roll back to some previous choice point $P'\mid Q'$, 
	and to try a different branch of the interaction;
\item 
	or, when checking the compliance of the specifications,
	it is statically computed which are the safe choices, if any, {\em before} running $P \mid Q$, so that they can be stored into
	a mediating process.
\end{itemize}
The reason for preferring the second approach is clearly apparent, as it limits the backtracking to  static type checking, while avoiding
it at run-time. In fact, once types of opposite end-points have been recognized compliant
 up to retractability of certain choices (that are kept distinct syntactically from unretractable ones),
and an orchestrator $\orchF$ has been sinthesized, the very same orchestrator can serve as guidance in the interaction on the session channel. 
In particular, by putting the orchestrator $\orchF$ in parallel with the processes holding the end-points, it can be used to drive the proper choices at run time.
If $k$ is the session channel, we write the resulting session by: 
\[(\nu k)(\namedorch{k}\orchF \mid P \mid Q)\]
representing that the interaction over the private session channel $k$ is ruled by $\orchF$.\\
The orchestrator is obviously an abstraction that allows for many different implementations.
Among the possible ones there are:
\begin{enumerate}
\item a communication infrastructure responsible of running the orchestrator;
\item  the two partners agree on an orchestrator that it is used on both sides as a communication interface.
\end{enumerate}
We do not further discuss implementation issues that would lead us outside the scope of the 
present paper.


%

{\bf Speculative selection.}
In session-types formalism the type $\selectiontype{l_1{:}S_1,\ldots,l_n{:}S_n}$ describes the protocol  consisting in selecting the label $l_i$ to be sent as output, and continuing as specified in $S_i$; this is the consequence of an internal choice that is transparent to the other participant in the session, that is expected to be able to react to all $l_1, \ldots , l_n$. The dual is the branching type  $\branchtype{l_1{:}S'_1,\ldots,l_n{:}S'_n}$, expecting a label among $l_1, \ldots , l_n$ as input to continue according to the respective continuation.
To these we add a new type constructor written
$$
\retrseltype{l_1{:}S_1,\ldots,l_n{:}S_n}
$$
that we dub {\em speculative selection type}. The intended meaning of speculative selection is: try selecting labels among $l_1, \ldots , l_n$
until an $l_i$ is found such that $l_i:S'_i$ is in the corresponding branching type, and $S_i$ and $S'_i$ are compliant. Observe that the
speculative selection type has no dual, so that we cannot use the notion of duality in the system.

As a running example suppose that 
a {\em Client} is willing to establish a session with a movie-{\em Provider} and to behave 
on her channel end according to the following session type: 

\medskip
\centerline{
$\mathsf{ClientSess}\ =\ \ \outputtype{\texttt{String}}.\selectiontype{\text{\sc buy}{:}\selectiontype{\text{\sc uhd}{:}\mathsf{S}, \text{\sc hd} {:}\mathsf{S}},\ \text{\sc rent}{:}\retrseltype{\text{\sc uhd}{:}\mathsf{S},\text{\sc hd}{:}\mathsf{S}, \text{\sc sd} {:}\mathsf{S} , \text{\sc ld} {:}\mathsf{S}} }$
}

\medskip
\noindent
where
$\mathsf{S}\ =\ \outputtype{\texttt{String}}.\branchtype{\text{\sc ok}{:}\inputtype{\texttt{Url}}, \text{\sc no}{:}\stopA}$. 
Accordingly the client will send on the session channel a login information (an element of the ground type {\tt String}) and then will internally 
decide whether she intends to {\sc buy} or to {\sc rent} a movie.
In the former case, she will further decide whether to buy an ultra-high-definition ($\text{\sc uhd}$) or a high-definition ($\text{\sc hd}$) movie.
In the latter case, instead, it is stated in the speculative selection type
$\retrseltype{\text{\sc uhd}{:}\mathsf{S},\text{\sc hd}{:}\mathsf{S}, \text{\sc sd} {:}\mathsf{S} , \text{\sc ld} {:}\mathsf{S}}$  that she will proceed according to four possible failure-amenable choices: renting an ultra-high-definition ($\text{\sc uhd}$), a high-definition ($\text{\sc hd}$), a standard-definition ($\text{\sc sd}$) or a low-definition ($\text{\sc ld}$) movie.
In all cases, she  will proceed according to type $\mathsf{S}$ by sending the string with the title of the movie, and either
receiving the URL from which the movie can be downloaded, if available (availability corresponds to the reception of {\sc ok}), or ending the 
session if it is not ({\sc no}).

From the previous discussion, if the {\em Client}  behaves on her end-point of the session channel 
according to the type $\mathsf{ClientSess}$, she
can safely interact with the {\em Provider} in case the latter behaves on the other end-point of the session channel according to
the following session type $\mathsf{ProvSess}$:

\medskip
\centerline{
$\mathsf{ProvSess}\ =\ \ \inputtype{\texttt{String}}.\branchtype{\text{\sc buy}{:}\branchtype{\text{\sc uhd}{:}\mathsf{S'}, \text{\sc hd} {:}\mathsf{S'}},\ \text{\sc rent}{:}\branchtype{\text{\sc hd}{:}\inputtype{\texttt{Nat}}.\mathsf{S'}, \text{\sc sd} {:}\mathsf{S'} , \text{\sc ld} {:}\mathsf{S'}} }$
}

\medskip\noindent
where
$\mathsf{S'}\ =\ \inputtype{\texttt{String}}.\selectiontype{\text{\sc ok}{:}\outputtype{\texttt{URL}}, \text{\sc no}{:}\stopA}$.

Now this client/server interaction succeeds in case the client actually buys a movie.
If the client intends to rent a movie, instead,  the choice {\sc uhd} will produce a synchronization failure, 
since no ultra-high-definition movies are for rent on that server. Also the choice {\sc hd} of
a movie to be rent would produce a syncronization failure, since then
the server requires a membership code 
(as described by $\text{\sc hd}{:}\inputtype{\texttt{Nat}}.\mathsf{S'}$). 
But this is not all the story, as the remaining two possibilities still lead to success, so that we insist that these two participants can agree at least
in part. Indeed to prevent synchronization failures the  {\em Client} will be instructed at run-time by the orchestrator to select either 
{\sc sd} or {\sc ld} choice only while renting a movie. The orchestrator is computed
when the session between the client and the server is tried and possibly opened, hence before it is started.
An orchestrator for the example is:

\medskip
\centerline{
$\mathsf{f}\ =\ \io .(\,(\text{\sc buy}.(\text{\sc uhd}.\mathsf{f'}+\text{\sc hd}.\mathsf{f'}))+(\text{\sc rent}.(\text{\sc sd}.\mathsf{f'}\oplus\text{\sc ld}.\mathsf{f'}))\,)$
}

\medskip\noindent
where $\mathsf{f'}\ =\ \io .(\,(\text{\sc ok}.\io)+\text{\sc no}\,)$.
This means that on the session channel between {\em Client} and {\em Provider}, the orchestrator $\mathsf{f}$ first enables an input/output interaction ( $\io$ ) and then either ($+$) a  $\text{\sc buy}$ or a $\text{\sc rent}$ branching.
In case of $\text{\sc rent}$, the orchestrator internally decides ($\oplus$) to force
either an $\text{\sc sd}$ or $\text{\sc ld}$ choice; this can be left open as both of them are ``safe'', namely do not
lead to (not even future) synchronization failures.

Observe that also the following two orchestrators can successfully drive the interaction
between {\em Client} and {\em Provider}:
\[\mathsf{f_1}\ =\ \io .(\,(\text{\sc buy}.(\text{\sc uhd}.\mathsf{f'}+\text{\sc hd}.\mathsf{f'}))+\text{\sc rent}.\text{\sc sd}.\mathsf{f'}\,), \quad
\mathsf{f_2}\ =\ \io .(\,(\text{\sc buy}.(\text{\sc uhd}.\mathsf{f'}+\text{\sc hd}.\mathsf{f'}))+\text{\sc rent}.\text{\sc ld}.\mathsf{f'}\,).\]
As a matter of fact, our calculus and its type system do consider any safe orchestrator for a session interaction. 
In actual implementations, however, one could be interested in limiting the nondeterminism in a session exclusively to that
exposed by the two partners. This involves considering {\em deterministic} orchestrators only,
like $\mathsf{f_1}$ and $\mathsf{f_2}$ above, that is with no $\oplus$ inside.

{\bf Adding priorities.}
For the semantics of the type 
$\retrseltype{\text{\sc uhd}{:}\mathsf{S},\text{\sc hd}{:}\mathsf{S}, \text{\sc sd} {:}\mathsf{S} , \text{\sc ld} {:}\mathsf{S}}$ the actual ordering of the labels is immaterial; therefore the actual choice stored in the orchestrator is up to the synthesis algorithm.
In practice, however, the choice of either $\mathsf{f_1}$ or $\mathsf{f_2}$ should not be randomly determined, rather it might reflect a particular policy representable at type level. A simple way is to provide a priority ordering among the failure-amenable choices, that can be expressed by
the following modified syntax:

\medskip
\centerline{
$\prtyretrseltype{\text{\sc ld}{:}\mathsf{S},\text{\sc sd}{:}\mathsf{S}, \text{\sc hd} {:}\mathsf{S} , \text{\sc uhd} {:}\mathsf{S}}$
}

\medskip\noindent
where $\!\, \langle\!\langle \hspace{1mm} \rangle\!\rangle$ is now an ordered list.
So the above type expresses that the option \text{\sc ld} is the one liked best, but in case of failure, \text{\sc sd} is the second preferred choice, and so on.  The priorities represented by the above speculative selection type force  the orchestrator $\mathsf{f_2}$ to
be synthesized.

We shall show that it is possible to synthesize the 
deterministic orchestrator which reflects the priorities described in the speculative selection types.
Moreover,
with an extra computational effort,
the priority-ordering policy can be lowered down and described at process level. 


According to the compliance relation, in a session only the client's communicating actions are guaranteed to be
matched by corresponding actions on the server's side. Therefore there might be some pending communications
on the server side even in case the client has successfully completed. 
This fact produces particular deadlocks which do not show up in ordinary session-based calculi and type systems, as these are
based on the notion of duality.
We show how to get rid of such stuck states by adding suitable reduction rules.
The type system can in fact be proved to be sound w.r.t. the new reductions, preventing typed process from reaching these peculiar stuck states.

In order to focus on the main concepts, in the present paper we do not treat recursion, that can be easily added in a fairly standard way.\\

{\bf Structure of the paper.} In Section \ref{sec:orchComply} we define types, orchestrators and the
relation of orchestrated compliance. The syntax of the calculus and the type system are treated
in Section \ref{sec:calculusTypes}. In Section \ref{sec:opsem} the operational semantics is defined
and the subject reduction property is proved, obtaining error freeness of typed processes as a corollary.
The particular deadlocks due to the use of the compliance relation instead of  
duality is dealt with in Section \ref{sec:cleanup}, whereas in Section \ref{sec:priorities}
we restrict orchestrators to deterministic ones, and make them implement 
a priority selection policy which can be described either at type or process level. 
Section \ref{sec:conclusions} contains the conclusions and suggests possible extensions.

\section{Session types and orchestrated compliance} \label{sec:orchComply}

First we introduce session types following \cite{HondaVK98} but for recursion (omitted for sake of simplicity) and for the new type
$\retrseltype{l_i{:}S_i}_{i\in I}$ for speculative selection, corresponding to retractable choice in \cite{BDLdL16}.

\begin{definition}[Types]
\label{def:types}
\[\begin{array}{lcl@{\hspace{6mm}}l@{\hspace{12mm}}l@{\hspace{6mm}}l@{\hspace{4mm}}l}
G &:=& \multicolumn{2}{l}{\mathtt{Nat} ~\mid~\mathtt{Bool}  ~\mid~ \ldots}  & & \mathrm{ground\    types}\\[2mm]
S &:=& & & & \mathrm{session\ types} \\ [1mm]
&   & \mid~ \stopA  & \text{terminated}& & \\
   &   & \mid~ \inputtype{G}.S &  \text{value input}  &
           \mid~  \outputtype{G}.S & \text{value output}\\
   &   & \mid~ \inputtype{S_1^p}.S_2 & \text{session input} &
           \mid~  \outputtype{S_1^p}.S_2 & \text{session output} 
\\
   &    &\mid~ \selectiontype{l_i{:}S_i}_{i\in I} & \text{selection}&
           \mid~  \branchtype{l_i{:}S_i}_{i\in I} & \text{branching}\\ 
   &   & \mid~  \retrseltype{l_i{:}S_i}_{i\in I} &\multicolumn{3}{l}{\hspace{-2mm}\text{speculative selection} } \\	[2mm]
T &:=& G ~\mid~ S^p & & & \textsc{i/o}\ \mathrm{types}
\end{array}
\]
where $p\in\Set{+,-}$.
Moreover, the labels $l_i$'s belong to a countable set of labels $\Labels$ and are pairwise distinct
in branching, selection and speculative-selection types.
\end{definition}

The syntax of {\em orchestrators} is inspired by that in \cite{Padovani10} and \cite{BvBdL17}.
Main differences with respect to \cite{Padovani10} and \cite{BvBdL17} are that we do not
have buffers, here unnecessary because we do not model asynchronous communications;
besides, our orchestrators can introduce some nondeterminism in orchestrated interactions
(see Definition \ref{def:opSem}).

\begin{definition}[Orchestrators]
We define the set $\Orch$ of {\em orchestrators}, ranged over by $\orchF, \orchG, \ldots$, as the terms generated by the following grammar:
 \[ \begin{array}{lrl@{\hspace{4mm}}l@{\hspace{10mm}}l}
\orchF, \orchG & ::= & \stopf  & \mbox{idle}\\
 & \mid & \io.\orchF & \mbox{{\sc i/o} prefix} 
\\
 & \mid  & l.\orchF & \mbox{selection prefix} \\
 & \mid & l.\orchF + l'.\orchG \quad (l \neq l') & \mbox{external choice}\\
 & \mid & l.\orchF \oplus l'.\orchG \quad (l \neq l') & \mbox{internal choice}\\
\end{array} \] 
where $l,l'\in\Labels$.
\end{definition}
\noindent
We write $\orchPlus_{i\in I}l_i.\orchF_i$ $($resp. $\orchOplus_{i\in I}l_i.\orchF_i)$   for $l_1.\orchF_1 + \cdots + l_n.\orchF_n$ $($resp. $l_1.\orchF_1 \oplus \cdots \oplus l_n.\orchF_n$), where $I = \Set{1, \ldots, n} \neq\emptyset$ and
the $l_i$'s are pairwise distinct. If $I$ is a singleton then $\orchPlus_{i\in I}l_i.\orchF_i$ ($\orchOplus_{i\in I}l_i.\orchF_i$)  is just a selection prefix.

\begin{definition} [Orchestrated compliance]
\label{def:orchCompliance}
The relation $\orchComply{\orchF}{S}{S'}$ among the orchestrator $\orchF$ and session types $S, S'$ is the least one such that:
\begin{enumerate}
\item
\label{def:orchCompliance-i}
 $\stopf : \stopA \comply S$, for any $S$,
\item if $\ \orchComply{\orchF}{S}{S'}$ then  $\orchComply{\io.\orchF}{ \inputtype{G}.S}{\outputtype{G}.S'}$ and 
	$\orchComply{\io.\orchF }{\outputtype{G}.S}{\inputtype{G}.S'}$ for any $G$,
\item if $\ \orchComply{\orchF_i}{S_i}{S'_i}$ for all $i \in I$ then, for any set of indexes $J$, \\
	$\orchComply{\orchPlus_{i\in I}l_i.\orchF_i}{\selectiontype{l_i{:}S_i}_{i\in I}} {\branchtype{l_j{:}S'_j}_{j\in I \cup J}}$ and 
	$\orchComply{\orchPlus_{i\in I}l_i.\orchF_i}{\branchtype{l_j{:}S_j}_{j\in I \cup J}}{\selectiontype{l_i{:}S'_i}_{i\in I}}$,
\item \label{def:orchCompliance-iv}
	if  \ $H \subseteq I\cap J$ with $H \neq \emptyset$ and $\orchComply{\orchF_h}{S_h}{S'_h}$ for all $h \in H$ then \\
	$ \orchComply{\orchOplus_{h\in H}l_h.\orchF_h }{\retrseltype{l_i{:}S_i}_{i\in I} }{ \branchtype{l_j{:}S'_j}_{j\in J}}$ and 
	$ \orchComply {\orchOplus_{h\in H}l_h.\orchF_h}{\branchtype{l_j{:}S_j}_{j\in J}}{ \retrseltype{l_i{:}S'_i}_{i\in I}}$.
\end{enumerate}
We say that $S$ and $S'$ are {\bf compliant}, written $S \comply S'$, if $\orchComply{\orchF}{S}{S'}$ for some $\orchF$.
\end{definition}

\begin{proposition} \label{prop:compOrchestrator}
Given $S, S'$ it is decidable whether $S \comply S'$. Moreover if $S \comply S'$ then an orchestrator $\orchF$ such that
$\orchComply{\orchF}{S}{S'}$ is computable.
\end{proposition}
\begin{proof} By induction over the structure of $S,S'$.
\end{proof}

\begin{remark} \label{rem:compOrchestrator}
Both Definition \ref{def:orchCompliance} and Proposition \ref{prop:compOrchestrator} easily extend to
the case of (contractive) recursive types and orchestrators
(see \cite{BvBdL17}, where orchestrated compliance is defined for ``session contracts'' instead of types). 
\end{remark}

\begin{remark}
\label{rem:dualityuptosub}
In the theory of session contracts \cite{BdL13,BH16}, compliance does correspond to the composition of duality and subtyping. Such a correspondence does transfer also to session types
in a very general sense, as shown in \cite{BDGK14}. In our setting, even if a relation of subtyping can be obtained out of a restriction of the subcontract relation defined in \cite{BLdL17}, the above mentioned correspondence looks unrealistic since, as    
pointed out in \cite{BDLdL16}, there exists no natural notion of duality\footnote{Duality for retractable session contracts can be immediately recovered by extending the formalism of \cite{BDLdL16} with speculative input choices, as done in \cite{BLdL17}. Such extension, however, do not seem to have a clear session-type counterpart.} for retractable contracts. 
\end{remark}

\begin{example}
{\em 
Extending the running example in the Introduction, we include the
higher-order features of session types.
We assume that the \textit{Client} of the movie-\textit{Provider} has to pay for the buyed/rented movie.
For the payment, if the movie is available, the \textit{Provider} throws to the \textit{Client} a session channel (typed by the session $\mathsf{PAY}$) that can be used to pay using several possible cards.
The new versions of $\mathsf{ClntSess}$ and $\mathsf{ProvSess}$ are now, respectively:\\[1mm]
\centerline{
$\mathsf{ClntSess}\  =\  \outputtype{\texttt{String}}.\selectiontype{\text{\sc buy}{:}\selectiontype{\text{\sc uhd}{:}\mathsf{S}, \text{\sc hd} {:}\mathsf{S}},\ \text{\sc rent}{:}\retrseltype{\text{\sc uhd}{:}\mathsf{S},\text{\sc hd}{:}\mathsf{S}, \text{\sc sd} {:}\mathsf{S} , \text{\sc ld} {:}\mathsf{S}} }
$}
{where} \hspace{1mm}
$\mathsf{S}\  =\  \outputtype{\texttt{String}}.\selectiontype{\text{\sc ok}{:}\inputtype{ \mathsf{PAY}^- }.\inputtype{\texttt{Url}}, \text{\sc no}{:}\stopA}
$\\
and\\
\centerline{
$\mathsf{ProvSess}\ =\ 
 \inputtype{\texttt{String}}.\branchtype{\text{\sc buy}{:}\branchtype{\text{\sc uhd}{:}\mathsf{S'}, \text{\sc hd} {:}\mathsf{S'}},\ \text{\sc rent}{:}\branchtype{\text{\sc hd}{:}\inputtype{\texttt{Nat}}.\mathsf{S'}, \text{\sc sd} {:}\mathsf{S'} , \text{\sc ld} {:}\mathsf{S'}} }
$}
{where} \hspace{1mm}
$\mathsf{S'}\ =\ 
\inputtype{\texttt{String}}.\selectiontype{\text{\sc ok}{:}\outputtype{ \mathsf{PAY}^- }.\outputtype{\texttt{Url}}, \text{\sc no}{:}\stopA}$

Here we assume that the payment always succeedes and the \textit{Client} is not cheating.
The safe interaction between \textit{Client} and the \textit{Provider} is guaranteed by 
$\mathsf{ClntSess}\comply\mathsf{ProvSess}$.
In fact, for instance, it can be checked that:

\centerline{$\mathsf{g}:\mathsf{ClntSess}\comply\mathsf{ProvSess}$}

\noindent
where
$
\mathsf{g}\ =\ \io .(\,(\text{\sc buy}.(\text{\sc uhd}.\mathsf{g'}+\text{\sc hd}.\mathsf{g'}))+(\text{\sc rent}.(\text{\sc sd}.\mathsf{g'}\oplus\text{\sc ld}.\mathsf{g'}))\,)
$
with $\mathsf{g'}\ =\ \io .(\,(\text{\sc ok}.\io.\io)+\text{\sc no}\,)$.

For the payment the \textit{Provider} establishes a session with the \textit{Bank}, hence acting as a client. The \textit{Provider}'s end of the corresponding session
channel is typed by:\\[1mm]
\centerline{
$ \mathsf{bankCustSess}\ 
  =\   \outputtype{\texttt{Amount}}.\mathsf{PAY}$
}
{where}\hspace{1mm}
$\mathsf{PAY}\ =\   \retrseltype{\text{\sc Diners}{:}\outputtype{\mathtt{ccNumber}} ,\text{\sc M\!card}{:}\outputtype{\mathtt{ccNumber}}, \text{\sc Visa}{:}\outputtype{\mathtt{ccNumber}} }
$\\
(notice that $\mathsf{PAY}$ is precisely the type of the channel-end that the \textit{Client}
receives from the movie-\textit{Provider} during the interaction described by $\mathsf{ClntSess}$).

Hence, once the \textit{Provider}-\textit{Bank} session is established, the \textit{Provider} sends  the cost of the movie and then delegates the actual payment to the \textit{Client}.
The polarity `$-$' in the channel-send of $\mathsf{ProvSess}$ indicates that the delegated channel's end is sent by the applicant of the session. Therefore the receiver has to act as such
\footnote{The polarities of two corresponding channel-send and channel-receive are always the same, since they both refer to the same channel's end; in particular, in the present example the polarity is `$-$' also in the channel-receive of $\mathsf{ClntSess}$.}. \\
On the delegated channel the \textit{Client} is allowed to pay either by using a {\sc Diners}, a {\sc Mastercard} or a {\sc Visa}, but it is not guaranteed that all of them will be actually available.

We also assume that once the \textit{Provider}-\textit{Bank} session is established, the \textit{Bank} behaves on its channel's end according to the type \\[1mm]
\centerline{
$ \mathsf{BankSess}\  =\  
\inputtype{\texttt{Amount}}.\branchtype{\text{\sc Discover}{:}\mathsf{bS} ,\text{\sc M\!card}{:}\mathsf{bS} ,
                                 \text{\sc Visa}{:}\mathsf{bS} ,\text{\sc \!A\!Expr}{:}\mathsf{bS} }$
}
{where}
$\mathsf{bS}\ =\  \inputtype{\mathtt{ccNumber}}.\outputtype{\mathtt{TransIDnum}}$.
Observe that the \textit{Bank} can accept any of the $\text{\sc Discover}$, $\text{\sc M\!card}$,
                                 $\text{\sc Visa}$ or $\text{\sc \!A\!Expr}$ cards, but does not accept
$\text{\sc Diners}$. Moreover, after receiving the credit-card number, the \textit{Bank} does issue the identifier of the transaction
(an element of the ground type \texttt{TransIDnum}), notwhithstanding it is not requested by
the $\mathsf{bankCustSess}$ session type. Nonetheless the safeness of the interaction with the \textit{Bank}
is guaranteed by the fact that 
{$\ \mathsf{bankCustSess} \comply \mathsf{BankSess}$}.
In fact, for instance, it can be checked that\\[1mm]
\centerline{$\mathsf{h}:\mathsf{bankCustSess}\comply\mathsf{BankSess}$}
where
$
\mathsf{h}\ =\ \io .(\,(\text{\sc M\!card}.\io) \oplus (\text{\sc Visa}.\io)\,).
$

}
\end{example}

\section{Calculus and type assignment} \label{sec:calculusTypes}

We assume to have a countable set ${\cal K}$ of {\em channel names}, ranged over by $k,k',\ldots$
In a session, we refer to the process performing the session-opening {\tt request} as ``the client'', while the process which {\tt accept}s it is referred to as ``the server''.

We distinguish among user-defined and run-time processes. While the former represent the code of concurrent programs, the latter
formalise the state of the system at run-time.
Following \cite{GH05, YoshidaV07}, clients' and servers' channel-ends
are identified by means of polarities $-$ and $+$ respectively, ranged over by $p,q,\ldots$ Moreover,
we define $k^{\Dual{+}} = k^-$ and $k^{\Dual{-}} = k^+$.

\begin{definition}[Processes]
\label{def:processes}
The set of (user-defined) \emph{processes} is defined as the set of the \emph{closed} expressions generated by the following 
grammar: 
\[\begin{array}{lcl@{\hspace{5mm}}l@{\hspace{12mm}}l@{\hspace{5mm}}l}
P, Q&:=& \mid ~ \inact &\mbox{terminated} & \mid ~ (\ P\ \mid\ Q\  ) &   \mbox{parallel} \\[.5mm]
       &     & \mid ~\typedrequest{a}{S}{k}{P}  & \mbox{session request} &
       \mid ~\typedaccept{a}{S}{k}{P}& \mbox{session accept}\\[.5mm]
       &     & \mid ~\send{k}{e}.P & \mbox{value send} &      
       \mid ~\receive{k}{x}.P& \mbox{value receive}\\[.5mm]
 &     & \mid ~\throw{k}{k'}{P} & \mbox{channel send} &      
       \mid ~\catch{k}{k'}{P}& \mbox{channel receive}\\[.5mm]

       &     & \mid ~\select{k}{l}.P &      \mbox{selection} &
       \mid ~\branch{k}{\Set{l_i{:}P_i}_{i\in I}}& \mbox{branching}\\[.5mm]
&     & \mid ~\retrselop{k}{l_i{:}P_i}_{i\in I} & \multicolumn{3}{l}{\!\!\! \mbox{speculative selection}} 
\end{array}
\]
where $k,k'\in{\cal K}$, $I$ is non-empty and finite, and the labels in branching and speculative selection, all belonging to $\Labels$ (as the one in selection), are pairwise distinct.

Let $P$ be any user defined process. {\em Run-time processes} are defined by the grammar
\[\begin{array}{lll}
R & ::= (\nu k)(\namedorch{k}{\orchF} \mid P) ~\mid~ (\nu k)(\namedorch{k}{\orchF} \mid R)  ~\mid~ (P \mid R) ~\mid~ (R \mid P) ~\mid~ (R\mid R')
\end{array}\]
where channels can be {\em polarized channels}, i.e. channel names can be decorated with polarities.\\
In $(\nu k)(\namedorch{k}{\orchF} \mid P)$ and $(\nu k)(\namedorch{k}{\orchF} \mid R)$ the operator $(\nu k)$ is a binder of $k^p$, $k^{\Dual{p}}$ and $k$. 

\end{definition}

In session type systems the compatibility of processes establishing a session lays on the notion of duality of the respective types,
that are associated to a port name which is the same on both sides. 
To duality we have replaced the notion of compliance. It is hence not straightforward to use port names for establishing a session between a client and a suitable server, since there is not 
a unique type for the possible servers of a given client.
So, in order to get a calculus as general and simple as possible,  we equip the {\tt request} and {\tt accept} operators with their types,  without any reference to port names.


Most process actions are similar to those of the calculus in \cite{HondaVK98},
but for the new {\em speculative selection}. A process $\retrselop{k}{l_i{:}P_i}_{i\in I}$
is able to send on channel $k$ any of the labels $l_i$'s and to proceed afterwards as $P_i$.
The process is aware that some of the synchronizations on the $l_i$'s could led to a failure
and an orchestrator is hence expected to drive the choice.\\

The processes \textit{Provider}, \textit{Client} and \textit{Bank} of the running example
can be described as follows (where also conditional processes are used and where only
the parts concerning rental are described).

\begin{example}{\em
\label{ex:processes}
Let \textbf{b} be a boolean expression representing the decision of whether buying or renting
a movie.\\[-2mm]
\centerline{
$\textit{Client}\ =\  \typedrequest{a}{\mathsf{ClntSess}}{k}{
\send{k}{\textrm{loginfo}}.
\IfThenElse{\textbf{b}}{\select{k}{\text{\sc buy}}.\textit{C}_{\text{\sc b}}}{\select{k}{\text{\sc rent}}.\textit{C}_{\text{\sc r}}}
}$
}
{where}\\
$
\textit{C}_{\text{\sc r}}\ =\ 
\retrselop{k}{\text{\sc uhd}.\textit{C}', \text{\sc hd}.\textit{C}',\text{\sc sd}.\textit{C}', \text{\sc ld}.\textit{P}' }$
\hspace{10mm}
$\textit{C}' \ =\ 
\send{k}{''zootropolis''}.\branch{k}{\Set{\text{\sc ok}{:\textit{C}'_{\text{\sc ok}} 
,\text{\sc no}{:}\inact }}}$\\
$\textit{C}'_{\text{\sc ok}}\ =\ \catch{k}{k'}
\retrselop{k'}{\text{\sc Diners}{:}Q,\text{\sc M\!card}{:}Q, \text{\sc Visa}{:}Q }$
\hspace{10mm}
$Q \ =\  \send{k'}{1234}.\receive{k}{url}.\textit{Watch}$\\[2mm]
\centerline{
$\textit{Provider}\ =\  \typedaccept{b}{\mathsf{ProvSess}}{k}{
\receive{k}{x}.
\branch{k}{\Set{\text{\sc buy}{:}\textit{P}_{\text{\sc b}},\text{\sc rent}{:}\textit{P}_{\text{\sc r}}}}}$}\\ 
{where}\\
$\textit{P}_{\text{\sc r}} \ = \
\branch{k}{\Set{\text{\sc uhd}.\textit{P}', \text{\sc hd}.\textit{P}',\text{\sc sd}.\textit{P}', \text{\sc ld}.\textit{P}'}}$
\hspace{8mm}
$\textit{P}' \ = \ 
\receive{k}{y}.\IfThenElse{\text{available(y)}}{\select{k}{\text{\sc ok}}.\textit{P}'_{\text{\sc ok}}} 
{\select{k}{\text{\sc no}}}$\\
$\textit{P}'_{\text{\sc ok}}\ = \
\typedrequest{c}{\mathsf{bnkCustSess}}{k'}{
\send{k'}{\text{amount(y)}}.
\throw{k}{k'} \send{k}{\text{url(y)}}
}$\\[2mm]
\centerline{
$\textit{Bank}\ =\ \typedaccept{d}{\mathsf{BankSess}}{k'}{
\receive{k'}{x}.
\branch{k'}{\Set{\text{\sc Discover}{:}B ,\text{\sc M\!card}{:}B ,
                                 \text{\sc Visa}{:}B ,\text{\sc \!A\!Expr}{:}B}}
}$}\\
{where}
$B  \ = \  \receive{k'}{cc}.\send{k'}{\text{IDtrans}(x,cc)}$
}\end{example}
 
\paragraph{The Type System}
The following type system is the ``more liberal'' system in \cite{YoshidaV07},
where duality is replaced by the relation of orchestrated compliance.
There are two kinds of judgments: the first one is $\Gamma \vdash e:G$, 
where $e$ is an expression of ground type;
the second one is $\Gamma \bvdash P \oftype \Delta$, where $P$ is either a user-defined or a run-time process. The {\em context} $\Gamma$ is a finite set of typings $x:G$ of expression variables;
 the {\em process-typing} (henceforth just {\em typing}) $\Delta$ is a finite set of typings $k^p:S$, where $k^p$ is a polarised channel name, and $S$ a session type. By $\dom(\Gamma)$ and $\dom(\Delta)$ we mean the set of variables or channel names that are typed, respectively, in $\Gamma$ and $\Delta$ . Variables and (polarized) channel names are pairwise distinct both in $\Gamma$ and $\Delta$ as usual; note that $k^+ \neq k^-$.

\begin{definition} [Type System] 
\label{def:typeSystem}
The rules of the type system are in Figure \ref{fig:typeSystem}.
In rule {\small $[$\textsc{Inact-T}$]$} a typing $\Delta$ is {\em completed} if for any $k^p \in \dom(\Delta)$ we have
$k^p:\stopA \in \Delta$.
In rule  {\small $[$\textsc{Conc-T}$]$} the typing $\Delta\cdot\Delta'$ is the union of the typings $\Delta$ and $\Delta'$ provided that
$\dom(\Delta) \cap \dom(\Delta') = \emptyset$, it is undefined otherwise; in the latter case the rule does not apply.
\end{definition}

The rules of Figure \ref{fig:typeSystem} are similar to those of 
\cite{YoshidaV07}, but for the following. In rules {\small [{\textsc{Acc-T}}]}, {\small [\textsc{Req-T}]} port names are not considered;
consequently contexts $\Gamma$ just contains expression variables and we do not have restrictions over port names. In \cite{YoshidaV07} rule {\small [\textsc{CRes-T}]} is
\[
\mathrule{\textsc{CRes-T}}
	{\Gamma \bvdash P \oftype \Delta \cdot k^-:S \cdot k^+:\Dual{S} }
	{\Gamma \bvdash (\nu k) P  \oftype \Delta}
\]
where $\Dual{S}$ is the dual of $S$. In the premise we have the typing $\Delta \cdot k^-:S \cdot k^+:S'$ instead, with the side condition
$\orchComply{\orchF}{S}{S'}$ yielding the orchestrator appearing in the process in the conclusion.


It is not difficult to check  the parallel composition of the processes of running example
is a typable user-defined process, in particular $\emptyset \bvdash \textit{Provider} \mid \textit{Client} \mid \textit{Bank} \oftype \emptyset$.

\begin{figure}[th]
\hrule
{\small
\[\begin{array}{c@{\quad}c}
\mathrule{\textsc{Inact-T}}
	{\Delta \text{ completed }}
	{\Gamma \bvdash  \inact \oftype \Delta}
	&
\mathrule{\textsc{Conc-T}}
	{\Gamma \bvdash P \oftype \Delta \quad \Gamma \bvdash Q \oftype \Delta' }
	{\Gamma \bvdash P \mid Q \oftype \Delta \cdot \Delta'}
\\ [8mm]
\mathrule{\textsc{Acc-T}}
	{\Gamma \bvdash P\Subst{k^+}{k} \oftype \Delta \cdot k^+:S}
	{\Gamma \bvdash \typedaccept{a}{S}{k}{P} \oftype \Delta}
&
\mathrule{\textsc{Req-T}}
	{\Gamma \bvdash P\Subst{k^-}{k} \oftype \Delta \cdot k^-:S}
	{\Gamma \bvdash \typedrequest{a}{S}{k}{P} \oftype \Delta}
\\ [8mm]
\mathrule{\textsc{Rec-T}}
 	{\Gamma, x:G \bvdash P \oftype \Delta\cdot k^p:S}
	{\Gamma \bvdash \receive{k^p}{x}.P \oftype \Delta\cdot k^p: \, \inputtype{G}.S}
&
\mathrule{\textsc{Send-T}}
	{\Gamma \vdash e:G \quad \Gamma \bvdash P \oftype \Delta\cdot k^p:S}
	{\Gamma \bvdash \send{k^p}{e}.P \oftype \Delta\cdot k^p: \, \outputtype{G}.S}
\\ [8mm]
\mathrule{\textsc{Cat-T}}
	{\Gamma \bvdash  P\Subst{k'^q}{k'} \oftype \Delta\cdot k^p:S_2\cdot {k'}^q:S_1}
	{\Gamma \bvdash   \catch{k^p}{{k'}}{P} \oftype \Delta\cdot k^p: \,\inputtype{S^q_1}.S_2 }
&
\mathrule{\textsc{Thr-T}}
	{\Gamma \bvdash  P \oftype \Delta\cdot k^p:S_2}
	{\Gamma \bvdash   \throw{k^p}{{k'}^q}{P} \oftype \Delta\cdot k^p: \, \outputtype{S^q_1}.S_2 \cdot {k'}^q:S_1}
\\[8mm]
\mathrule{\textsc{Br-T}}
	{\forall i\in I \supseteq J~~\Gamma \bvdash P_i \oftype \Delta\cdot k^p:S_i}
	{\Gamma \bvdash  \branch{k^p}{\Set{l_{i}:P_{i}}_{i\in I}} \oftype \Delta\cdot k^p:\branchtype{l_{j}{:}S_j}_{j\in J}}
&
\mathrule{\textsc{Sel-T}}
	{\Gamma \bvdash P \oftype \Delta\cdot k^p:S_j \quad j\in I}
	{\Gamma \bvdash  \select{k^p}{l_j}.P \oftype \Delta\cdot k^p:\selectiontype{l_i{:}S_i}_{i\in I}} 
\\[8mm]
\multicolumn{2}{c}{
\mathrule{\textsc{SSel-T}}
	{\forall i\in I~~\Gamma \bvdash P_i \oftype \Delta\cdot k^p:S_i}
	{\Gamma \bvdash  \retrselop{k^p}{l_{i}:P_{i}}_{i\in I} \oftype \Delta\cdot k^p:\retrseltype{l_{i}{:}S_i}_{i\in I}}
	}
\\[8mm]
\mathrule{\textsc{CRes-T}}
	{\Gamma \bvdash P \oftype \Delta \cdot k^-:S_1 \cdot k^+:S_2   \qquad \orchComply{\orchF}{S_1}{S_2} }
	{\Gamma \bvdash (\nu k)(\namedorch{k}{\orchF} \mid P ) \oftype \Delta}
&
\mathrule{\textsc{CRes'-T}}	
	{\Gamma \bvdash P \oftype \Delta \qquad k^+,k^- \not \in \dom(\Delta)}
	{\Gamma \bvdash (\nu k) P \oftype \Delta}


\end{array}\]}
\hrule
\caption{The type system.}
\label{fig:typeSystem}
\end{figure}


\section{Operational semantics and error freeness} \label{sec:opsem}
%
%
%
%
Because of the presence of orchestrators in session interactions, the notion of structural congruence
has to be handled with some extra care than in usual calculi with session-types.  

\begin{definition} [Structural Congruence] \label{def:run-time-strCongr}
The {\em structural congruence} $\congr$ is the least congruence over run-time processes 
such that:
\begin{enumerate}
\item $R \congr R'$ if $R'$ is obtained from $R$ by alphabetical change of bound channel names, avoiding name clashes,
\item \label{def:run-time-strCongr-ii}
	$(\nu k)(\namedorch{k}{\orchF} \mid Q \mid (\nu k')(\namedorch{k'}{\orchG} \mid Q' \mid R)) \congr
	(\nu k')(\namedorch{k'}{\orchG} \mid (\nu k)(\namedorch{k}{\orchF} \mid Q \mid Q' )\mid R)$ if $k\not \in \fc{R}, k' \not \in \fc{Q}$,
\item if $X,Y,Z$ are either user-defined or run-time processes that are not a named orchestrator, then:
	$(X \mid Y) \congr (Y \mid X)$ and $(X \mid (Y \mid Z)) \congr ((X \mid Y) \mid Z)$.
\end{enumerate}
%
\end{definition}

\begin{definition}[Operational semantics]
\label{def:opSem}
The operational semantics of processes is described by the reduction rules listed in Figure
\ref{fig:opSem}.
\end{definition}

\begin{figure}
\hrule
\[\begin{array}{l@{\hspace{0mm}}l}

 \mathax{\textsc{Link}}{\typedrequest{a}{S}{k}{P} \mid \typedaccept{b}{S'}{k}{Q}  ~\lts{}~  
 (\nu k)(\ \namedorch{k}{\orchF}\ \mid\ P\Subst{k^-}{k}\ \mid\ Q\Subst{k^+}{k}\ )
} & \mbox{if $\ \orchComply{\orchF}{S}{S'}$}
\\[2mm]

 \mathax{\textsc{OrchComm}}{
(\nu k)(\ \namedorch{k}{\io.\orchF}\ \mid\ \send{k^p}{e}.P\ \mid\ \receive{k^{\Dual{p}}}{x}.Q\ )
 ~\lts{}~  
(\nu k)(\ \namedorch{k}{\orchF}\ \mid\ P\ \mid\ Q\Subst{v}{x}\ )
} & \mbox{if $\ e \downarrow v$}
\\[4mm]

\multicolumn{2}{l}{
 \mathax{\textsc{OrchDeleg}}{
(\nu k)(\ \namedorch{k}{\io.\orchF}\ \mid\ \throw{k^p}{{k'}^q}{P}\ \mid\ \catch{k^{\Dual{p}}}{{k'}}{Q}\ )
 ~\lts{}~ 
(\nu k)(\ \namedorch{k}{\orchF}\  \mid P\ \mid Q\Subst{k'^q}{k'}\ )
}
}
\\[4mm]

\mathax{\textsc{OrchSel}}{
(\nu k)(\ \namedorch{k}{\orchPlus_{h\in H}l_h.\orchF_h}\ \mid\ \select{k^p}{l_c}.P\ \mid\ \branch{k^{\Dual{p}}}{\Set{l_i{:}Q_i}_{i\in I}}\ )
 ~\lts{}~
(\nu k)(\ \namedorch{k}{\orchF_c}\  \mid P\  \mid Q_c\ )} & \mbox{if $\ c\in H\cap I$}

\\[3mm]

\mathax{\textsc{OrchSSel}}{
(\nu k)(\ \namedorch{k}{\orchOplus_{h\in H}l_h.\orchF_h}\ \mid\
\retrselop{k^p}{l_j:P_j}_{j\in J}\ \mid\ \branch{k^{\Dual{p}}}{\Set{l_i{:}Q_i}_{i\in I}} \ )
 ~\lts{}~
(\nu k)(\ \namedorch{k}{\orchF_c}\ \mid\ P_c\ \mid\ Q_c\ )} & \mbox{if $\ c\in H{\cap}I{\cap} J$}
\\ [2mm]
%

%
%
\end{array}\]
\[
\begin{array}{l@{\hspace{8mm}}c@{\hspace{8mm}}r}
\mathrule{\textsc{Par}}{
P ~\lts{}~ P'
}{
P ~\mid~ Q ~\lts{}~ P' ~\mid~ Q}
&
\mathrule{\textsc{Scop}}{
P ~\lts{}~ P'
}{
(\nu k)P  ~\lts{}~ (\nu k)P'}
&
\mathrule{\textsc{Str}}{
Q\equiv P ~\lts{}~ P' \equiv Q'
}{
Q  ~\lts{}~ Q'}
\end{array}
\]
\hrule
\caption{Operational Semantics}\label{fig:opSem}
\end{figure}
This operational semantics extends that in \cite{HondaVK98} by taking into account 
the new operator and the necessity for interactions of being orchestrated.
\begin{description}
\item {\small [\textsc{Link}]} 
If the session type $S$ is compliant with $S'$ by means of the orchestrator $\textsf{f}$, the session-opening request $\typedrequest{a}{S}{k}{P}$ can be accepted by  $\typedaccept{b}{S'}{k}{Q}$. A new channel is created for the opened session. 
The `$-$' end is owned by the process who requested the opening (the client) whereas the
 `$+$' end is owned by the other one (the server). To connect the orchestrator
$\textsf{f}$ to the opened session, $\textsf{f}$ is labelled with
the channel name $k$. This forces its orchestration actions to act only on synchronizations over the channel $k$.
\item {\small[\textsc{OrchComm}]}, {\small[\textsc{OrchDeleg}]}, {\small[\textsc{OrchSel}]} 
In these rules the orchestrator enables the communication of a value,
the communication of a channel, and the selection of a label, respectively.
In rule {\small [\textsc{OrchComm}]}
the side condition $e\! \downarrow\! v$ reads: expression $e$ evaluates to the value $v$.
\item {\small[\textsc{OrchSSel}]}
In presence of a speculative selection, i.e. a number of choices possibly leading to synchronization failures,
the role of the orchestrator is to suggest one among the safe choices.
Notice that, in case the cardinality of $ H\cap I\cap J$ is strictly greater than one, this rule is nondeterministic.
 In actual implementation it is reasonable to expect the orchestrator
not to add nondeterminism to the system. We show in Section \ref{sec:priorities} how this can be obtained by interpreting $[l_j:P_j]_{j\in J}$ in $\retrselop{k^p}{l_j:P_j}_{j\in J}$ as a priority list
and how is it possible to synthesize, out of $S$ and $S'$ in {\small [\textsc{Link}]},
the orchestrator that suggests the safe choice possessing
 the highest priority, if any.
\item {\small [\textsc{Par}]}, {\small [\textsc{Scop}]}, {\small [\textsc{Str}]} These rules are standard.
\end{description}

\begin{figure}[th]
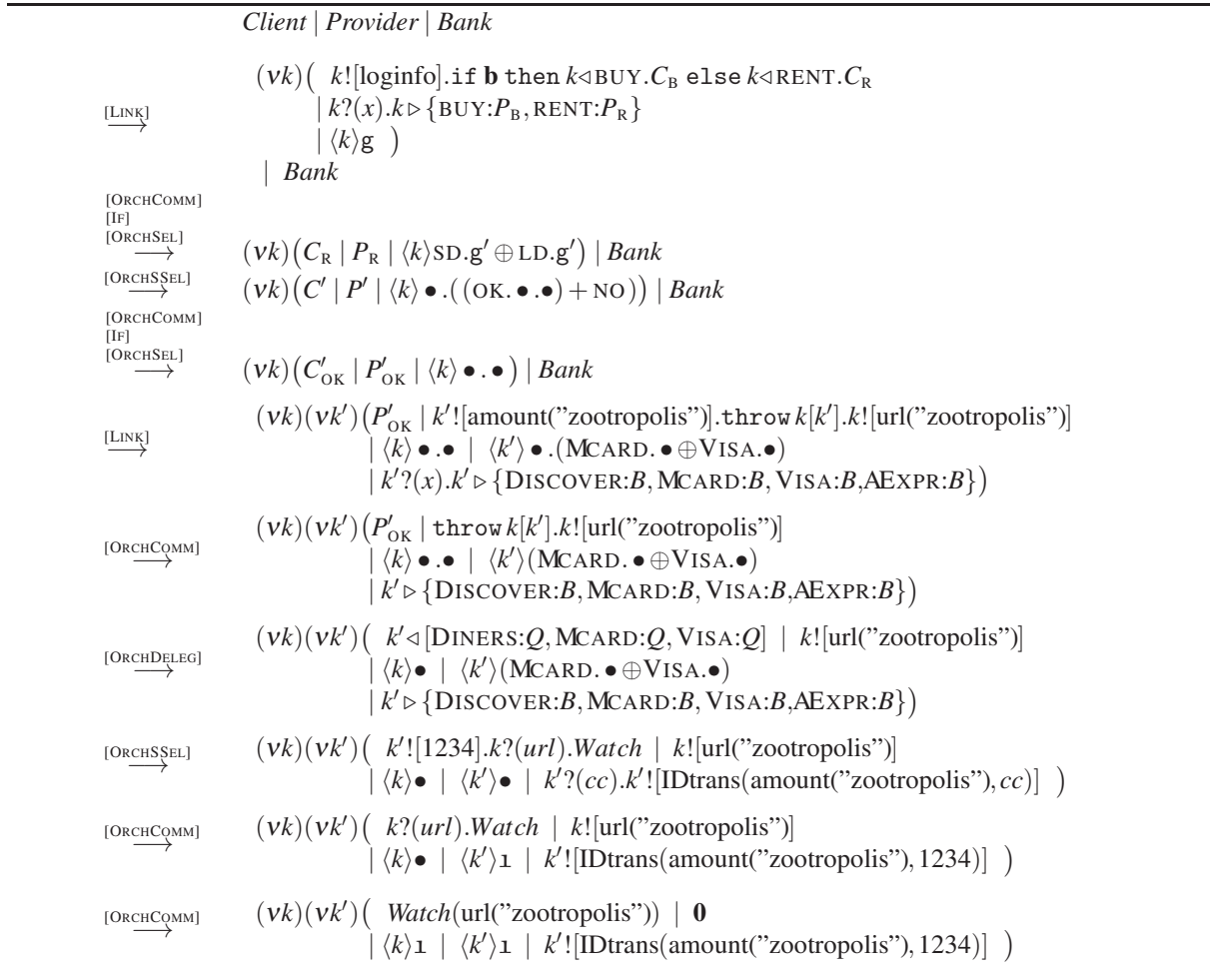

{\small
\hrule
\hspace{2mm}
\[\begin{array}{lll}
  &
\textit{Client} \mid \textit{Provider} \mid \textit{Bank}
&
	
\\[3mm]

\lts{\mbox{\tiny [\textsc{Link}]}}
&
\begin{array}{l}
 (\nu k)\big(\hspace{2mm}\send{k}{\textrm{loginfo}}.
\IfThenElse{\mathbf{b}}{\select{k}{\text{\sc buy}}.\textit{C}_{\text{\sc b}}}{\select{k}{\text{\sc rent}}.\textit{C}_{\text{\sc r}}}\\
\hspace{8mm}\mid \receive{k}{x}.
\branch{k}{\Set{\text{\sc buy}{:}\textit{P}_{\text{\sc b}},\text{\sc rent}{:}\textit{P}_{\text{\sc r}}}}\\
\hspace{8mm}\mid \namedorch{k}{\mathsf{g}}
\hspace{2mm} \big)\\
\hspace{1mm}\mid
\hspace{1mm} \textit{Bank}

\end{array}\\
\!\!\!\!\lts{\mbox{\tiny $\begin{array}{l}\text{[}\textsc{OrchComm}\text{]}\\ \text{[}\textsc{If}\text{]}\\ \text{[}\textsc{OrchSel}\text{]}\end{array}$}}
&
(\nu k)\big( \textit{C}_{\text{\sc r}} \mid \textit{P}_{\text{\sc r}} \mid \namedorch{k}{\mathsf{\text{\sc sd}.\mathsf{g'}\oplus\text{\sc ld}.\mathsf{g'}}}  \big) \mid \textit{Bank}\\

\lts{\mbox{\tiny [\textsc{OrchSSel}]}}
&
(\nu k)\big( \textit{C}' \mid \textit{P}' \mid \namedorch{k}{\io .(\,(\text{\sc ok}.\io.\io)+\text{\sc no}\,)}  \big) \mid \textit{Bank}\\
\vspace{2mm}

\!\!\!\!\lts{\mbox{\tiny $\begin{array}{l}\text{[}\textsc{OrchComm}\text{]}\\ \text{[}\textsc{If}\text{]}\\ \text{[}\textsc{OrchSel}\text{]}\end{array}$}}
&
(\nu k)\big( \textit{C}'_{\text{\sc ok}} \mid \textit{P}'_{\text{\sc ok}} \mid \namedorch{k}{\io.\io}  \big) \mid \textit{Bank}\\
\vspace{2mm}

\lts{\mbox{\tiny [\textsc{Link}]}}
&
\begin{array}{l}
(\nu k)(\nu k')\big( \textit{P}'_{\text{\sc ok}} \mid 
\send{k'}{\text{amount(''zootropolis'')}}.
\throw{k}{k'}  \send{k}{\text{url(''zootropolis'')}}\\
 \hspace{15mm}\mid \namedorch{k}{\io.\io}  \ \mid \
 \namedorch{k'}{\io.(\text{\sc M\!card}.\io\oplus \text{\sc Visa}.\io)}  \\
\hspace{15mm}\mid  \receive{k'}{x}.
\branch{k'}{\Set{\text{\sc Discover}{:}B ,\text{\sc M\!card}{:}B ,
                                 \text{\sc Visa}{:}B ,\text{\sc \!A\!Expr}{:}B}} \big) 
\end{array}\\
\vspace{2mm}

\lts{\mbox{\tiny [\textsc{OrchComm}]}}
&
\begin{array}{l}
(\nu k)(\nu k')\big( \textit{P}'_{\text{\sc ok}} \mid 
\throw{k}{k'}  \send{k}{\text{url(''zootropolis'')}}\\
 \hspace{15mm}\mid \namedorch{k}{\io.\io}  \ \mid \
\namedorch{k'}{(\text{\sc M\!card}.\io\oplus \text{\sc Visa}.\io)}  \\
\hspace{15mm}\mid  
\branch{k'}{\Set{\text{\sc Discover}{:}B ,\text{\sc M\!card}{:}B ,
                                 \text{\sc Visa}{:}B ,\text{\sc \!A\!Expr}{:}B}} \big) 
\end{array}\\
\vspace{2mm}

\lts{\mbox{\tiny [\textsc{OrchDeleg}]}}
&
\begin{array}{l}
(\nu k)(\nu k')\big( \hspace{2mm}\retrselop{k'}{\text{\sc Diners}{:}Q ,\text{\sc M\!card}{:}Q, \text{\sc Visa}{:}Q }\  \mid \
\send{k}{\text{url(''zootropolis'')}}\\
 \hspace{15mm}\mid \namedorch{k}{\io}  \
\mid\ \namedorch{k'}{(\text{\sc M\!card}.\io\oplus \text{\sc Visa}.\io)}  \\
\hspace{15mm}\mid  
\branch{k'}{\Set{\text{\sc Discover}{:}B, \text{\sc M\!card}{:}B, 
                                 \text{\sc Visa}{:}B, \text{\sc \!A\!Expr}{:}B}} \big) 
\end{array}\\
\vspace{2mm}

\lts{\mbox{\tiny [\textsc{OrchSSel}]}}
&
\begin{array}{l}
(\nu k)(\nu k')\big( \hspace{2mm}\send{k'}{1234}.\receive{k}{url}.Watch\
 \mid \
\send{k}{\text{url(''zootropolis'')}}\\
 \hspace{15mm}\mid \namedorch{k}{\io}\  \mid\ \namedorch{k'}{\io}\  \mid\
\receive{k'}{cc}.\send{k'}{\text{IDtrans}(\text{amount(''zootropolis'')},cc)} \hspace{2mm}\big) 
\end{array}\\
\vspace{2mm}

\lts{\mbox{\tiny [\textsc{OrchComm}]}}
&
\begin{array}{l}
(\nu k)(\nu k')\big( \hspace{2mm}\receive{k}{url}.Watch\ \mid \
\send{k}{\text{url(''zootropolis'')}}\\
 \hspace{15mm}\mid \namedorch{k}{\io}\  \mid\ \namedorch{k'}{\stopf}\  \mid\
\send{k'}{\text{IDtrans}(\text{amount(''zootropolis'')},1234)} \hspace{2mm}\big) 
\end{array}\\[3mm]

\lts{\mbox{\tiny [\textsc{OrchComm}]}}
&
\begin{array}{l}
(\nu k)(\nu k')\big( \hspace{2mm}\textit{Watch}(\text{url(''zootropolis'')})\  \mid \
\inact\\
 \hspace{15mm}\mid \namedorch{k}{\stopf}\  \mid\ \namedorch{k'}{\stopf}\  \mid\
\send{k'}{\text{IDtrans}(\text{amount(''zootropolis'')},1234)} \hspace{2mm}\big) 
\end{array}
\end{array}\]
}
\hrule
\caption{Reductions example}\label{fig:red-ex}
\end{figure}

\begin{example}{\em
We can now see in Figure \ref{fig:red-ex} the evolution of the user-defined process of our example
$$\textit{Client}\,\mid\,\textit{Provider}\,\mid\,\textit{Bank}$$
We assume that the client decides to rent an available  movie.
}\end{example}

\noindent
The type system guarantees type-checked processes to be free from (a version of) the standard synchronization errors of session-type-based calculi, which now involve also orchestrators
and that we dub {\em orchestrated synchronization errors}.
It also prevent (a class of) errors due to the absence of orchestration action and that we dub
{\em vacuous orchestration errors}.

\begin{example}
{\em 
\label{ex:errorexamples}
Let us see some examples of errors that cannot actually occur in typeable processes
(see Lemma \ref{lem:typablenoerr} below).
\begin{description}
\item 
$(\nu k)(\send{k^+}{e}.R \mid \send{k^-}{e'}.R' \mid \namedorch{k}{\orchF})$
This process is stuck. It cannot be typed since the types for $k^+$ and $k^-$ should have,
respectively, the form $\outputtype{G}.{S_1}$ and $\outputtype{G'}.{S_2}$. Rule $[\textsc{CRes-T}]$
cannot be applied to type the whole process, since, by Definition \ref{def:orchCompliance}, these types are not compliant (even in case $G=G'$).

\item
$(\nu k)(\select{k^+}{l}.R \mid \branch{k^+}{\Set{l_j{:}R_j}_{j\in J}} \mid \namedorch{k}{\orchPlus_{i\in I}l_i.\orchF_i})$
This stuck process cannot be typed since rule $[\textsc{CRes-T}]$ requires the compliant types to be assigned to polarized channels with different polarities.

\item 
$(\nu k)(\send{k^+}{e}.R \mid \receive{k^-}{x}.R' \mid \namedorch{k}{{\orchPlus_{i\in I}l_i.\orchF_i}})$
This process is stuck since the orchestrator does not enable the input/output syncronization.
It cannot be typed either. In fact the types for $k^+$ and $k^-$ should have,
respectively, the form $\outputtype{G}.{S_1}$ and $\inputtype{G}.{S_2}$, and rule $[\textsc{CRes-T}]$ cannot be applied since, by Definition \ref{def:orchCompliance},
it is impossible to have ${\orchPlus_{i\in I}l_i.\orchF_i}: \outputtype{G}.{S_1} \comply \inputtype{G}.{S_2}$.
\item 
$(\nu k)(\send{k^-}{e}.R \mid R' \mid \namedorch{k}{\stopf})$ Since any reduction is necessarily driven by an orchestrating action, there is no possibility for the  process 
$\send{k^-}{e}.R$ to progress. Unlike the previous examples, this sort of deadlock depends exclusively on the lack of orchestration actions.
\end{description}
}
\end{example}


\begin{definition}[Errors]\hfill
\begin{enumerate}[i)]
\item
A {\em $k^p$-process} is a run-time process term whose first action involves the channel $k$, namely a process having one of the 
following forms:

$\send{k^p}{e}.P$, $\ \receive{k^p}{x}.P$, $\ \throw{k^p}{{k'}^q}{P}$, $\ \catch{k^p}{{k'}^q}{P}$,
$\ \select{k^p}{l}.P$, $\ \branch{k^p}{\Set{l_i{:}P_i}_{i\in I}}$, $\ \retrselop{k^p}{l_i{:}P_i}_{i\in I}$;
\item A {\em potential $k$-redex} is a process term 
formed by the parallel composition of 
one  $k^p$-processe, one $k^q$-processe, for some $p$ and $q$, and one $k$-named orchestrator;
\item
A process $P$ is an {\em orchestrated synchronization error} (orch-synch error, for short) if it contains a potential {\em $k$-redex} which does not reduces;
\item
A process $R$ is a {\em vacuous-orchestration error} if it contains a subterm $R'$
such that\\
\centerline{
$R' = \namedorch{k}{\stopf} \mid R''$ ~~~~where $R''$ is a  $k^-$-process;}
\item
An {\em error} is either an orchestrated synchronization error or a vacuous-orchestration error.
\end{enumerate}
\end{definition} 

\noindent
The type system guarantees that
a typable process cannot be an error.\\

\begin{lemma}
\label{lem:typablenoerr}
Let $\Gamma \bvdash R \oftype \Delta$. Then $R$ is not an error.
\end{lemma}

By means of the Subject Reduction property we show that
a typable initial process never reduces to an error.\\

Concerning the typing of expressions, we assume the  standard property that if $\Gamma \vdash e:G$ and $e \downarrow v$ then
$\Gamma \vdash v:G$. Then the following technical lemma, to be used for the subject reduction property (Theorem \ref{thm:subjectReduction}), is proved by an easy induction on derivations:

\medskip
\begin{lemma}
\label{lem:subst}
If $\Gamma \vdash e:G$, $\Gamma, x:G \bvdash R \oftype \Delta$, and $e \downarrow v$ then
$\Gamma \bvdash R\Subst{v}{x} \oftype \Delta$.
\end{lemma}

Before proceeding with the Subject Reduction property, we show that typability is invariant with respect to the structural congruence formalized in Definition \ref{def:run-time-strCongr}.
The presence of orchestrators makes this usually fairly standard proof more subtle to handle. 
\begin{lemma}
\label{lem:subjcongr}
If $\Gamma \bvdash P \oftype \Delta$ and $P \congr  Q$ then
$\Gamma \bvdash Q \oftype \Delta$.
\end{lemma}

\begin{proof}
By induction over the definition of $\congr$. We illustrate the case \ref{def:run-time-strCongr}.\ref{def:run-time-strCongr-ii}, where we have:
\[(\nu k)(\namedorch{k}{\orchF} \mid P \mid (\nu k')(\namedorch{k'}{\orchG} \mid Q \mid R)) \congr
	(\nu k')(\namedorch{k'}{\orchG} \mid (\nu k)(\namedorch{k}{\orchF} \mid P \mid Q )\mid R)\]
and $k\not \in \fc{R}, k' \not \in \fc{P}$.
The derivation of $\Gamma \bvdash (\nu k)(\namedorch{k}{\orchF} \mid P \mid (\nu k')(\namedorch{k'}{\orchG} \mid Q \mid R)) \oftype \Delta$ 
ends by:
{\small
\[
\prooftree
	\prooftree
		\Gamma \bvdash P \oftype \Delta_2 
		\quad
		\prooftree
			\prooftree
				\Gamma \bvdash Q \oftype \Delta_5
				\quad
				\Gamma \bvdash R \oftype \Delta_6
			\justifies
				\Gamma \bvdash Q \mid R \oftype \Delta_4
			\endprooftree
			\quad \orchComply{\orchG}{S_3}{S_4}
		\justifies
			\Gamma \bvdash (\nu k')(\namedorch{k'}{\orchG} \mid Q \mid R) \oftype \Delta_3
		\endprooftree
	\justifies
		\Gamma \bvdash P \mid (\nu k')(\namedorch{k'}{\orchG} \mid Q \mid R) \oftype \Delta_1 
	\endprooftree
	\quad \orchComply{\orchF}{S_1}{S_2}
\justifies
	\Gamma \bvdash (\nu k)(\namedorch{k}{\orchF} \mid P \mid (\nu k')(\namedorch{k'}{\orchG} \mid Q \mid R)) \oftype \Delta
\endprooftree
\]
}
where $\Delta_1 = \Delta \cdot\; k^-:S_1\;\cdot\; k^+:S_2$ for some $S_1, S_2$;  $  \Delta_1=\Delta_2 \cdot \Delta_3$; 
$\Delta_4 = \Delta_3 \cdot\; k'^-:S_3\;\cdot\; k'^+:S_4$ for some $S_3, S_4$; finally $ \Delta_4=\Delta_5 \cdot \Delta_6 $.
Let's set $\Delta\setminus k = \Delta \setminus \bigcup_{S,S'}\Set{k^p:S, k^{\Dual{p}}:S'}$. Then such a derivation can be rearranged as follows:
{\small
\[
\prooftree
	\prooftree
		\prooftree
			\prooftree
				\Gamma \bvdash P \oftype \Delta_2
				\quad
				\Gamma \bvdash Q \oftype \Delta_5
			\justifies
				\Gamma \bvdash P \mid Q \oftype \Delta_2\cdot \Delta_5
			\endprooftree
			\quad \orchComply{\orchF}{S_1}{S_2}
		\justifies
			\Gamma \bvdash (\nu k)(\namedorch{k}{\orchF} \mid P \mid Q )  \oftype (\Delta_2\cdot \Delta_5) \setminus k
		\endprooftree
		\quad
		\Gamma \bvdash R \oftype \Delta_6
	\justifies
		\Gamma \bvdash (\nu k)(\namedorch{k}{\orchF} \mid P \mid Q )\mid R
			\oftype ((\Delta_2\cdot \Delta_5) \setminus k) \cdot \Delta_6
	\endprooftree
	\quad \orchComply{\orchG}{S_3}{S_4}
\justifies
	\Gamma \bvdash (\nu k')(\namedorch{k'}{\orchG} \mid (\nu k)(\namedorch{k}{\orchF} \mid P \mid Q )\mid R) 
		\oftype (((\Delta_2\cdot \Delta_5) \setminus k) \cdot \Delta_6) \setminus k'
\endprooftree
\]
}
To see that this is a correct derivation, let us assume that $k'\not \in \dom(\Delta_2)$, because $k' \not \in \fc{P}$ is the side condition
of case  \ref{def:run-time-strCongr}.\ref{def:run-time-strCongr-ii}.
Now $\Delta_2\cdot \Delta_5$ is defined since 
$\Delta_5\setminus k' = \Delta_3$, and we know that $\Delta_2 \cdot \Delta_3 = \Delta_1$ is defined. On the other hand
if $k'' \in \dom(\Delta_6)$ then $k'' \not \in \dom(\Delta_5)$ because $\Delta_5 \cdot \Delta_6 = \Delta_4$ is defined; also if
$k'' \neq k$ then $k'' \not \in \dom(\Delta_2)$ since $\dom(\Delta_6) \subseteq \dom(\Delta_3)\setminus \Set{k'}$ and
$\Delta_2$ and $\Delta_3$ are compatible. It follows that $((\Delta_2\cdot \Delta_5) \setminus k) \cdot \Delta_6$ is defined.
Finally, assuming without loss of generality that $k \not \in \dom(\Delta_6)$ as we know that $k\not \in \fc{R}$, we conclude:
\[\begin{array}{llll}
\Delta & = & (\Delta_2 \cdot (\Delta_5 \cdot \Delta_6)\setminus k')\setminus k & \mbox{by the first derivation}\\
& = & (\Delta_2 \cdot \Delta_5 \cdot \Delta_6) \setminus k' \setminus k & \mbox{since $k'\not \in \dom(\Delta_2)$} \\
& = & ((\Delta_2 \cdot \Delta_5) \setminus k ~\cdot~ \Delta_6 \setminus k) \setminus k' \\
& = & (((\Delta_2\cdot \Delta_5) \setminus k) ~\cdot~ \Delta_6) \setminus k' & \mbox{since $k\not \in \dom(\Delta_6)$} 
\end{array}\]
\end{proof}

\noindent
Since no typable processes is an orch-synch error  (Lemma \ref{lem:typablenoerr}),
the following theorem property guarantees (see Corollary \ref{cor:errorFree}) that no error can appear during the evolution of a typable user-defined process.

\begin{theorem}[Subject reduction]
\label{thm:subjectReduction}
If $\Gamma \bvdash P \oftype \Delta$ and $P \lts{}  Q$  then
$\Gamma \bvdash Q \oftype \Delta$.
\end{theorem}

\begin{proof}
By induction over the definition of $P \lts{}  Q$. In case of rule [\textsc{Link}] we have:
\[\typedrequest{a}{S}{k}{P} \mid \typedaccept{b}{S'}{k}{Q} \lts{}
		(\nu k)(\namedorch{k}{\orchF} \mid P\Subst{k^-}{k} \mid Q\Subst{k^+}{k} )
\]
for some $\orchF$ such that $\orchComply{\orchF}{S}{S'}$.
On the other hand by hypothesis and the shape of the rules, we have the derivation:
{\small \[
\prooftree
	\prooftree
		\Gamma \bvdash P\Subst{k^-}{k} \oftype \Delta\cdot k^-:S
	\justifies
		\Gamma \bvdash \typedrequest{a}{S}{k}{P} \oftype \Delta
	\endprooftree
	\quad
	\prooftree
		\Gamma \bvdash Q\Subst{k^+}{k} \oftype \Delta\cdot k^+:S'
	\justifies
		\Gamma \bvdash \typedaccept{b}{S'}{k}{Q} \oftype \Delta'
	\endprooftree
\justifies
	\Gamma \bvdash \typedrequest{a}{S}{k}{P} \mid 
		\typedaccept{b}{S'}{k}{Q} \oftype \Delta\cdot \Delta'
\endprooftree
\]
}
where $\Delta\cdot \Delta'$ has to be defined (namely $\dom(\Delta) \cap \dom(\Delta') = \emptyset$).
This implies that: 
\[(\Delta\cdot k^-:S) \cdot (\Delta'\cdot k^+:S') = \Delta\cdot \Delta' \cdot \; k^-:S \;\cdot \; k^+:S'
\] 
is defined as well, so that using the fact that $\orchF: S \comply S'$, we obtain the following derivation:
{\small 
\[
\prooftree
	\prooftree
		\Gamma \bvdash P\Subst{k^-}{k} \oftype \Delta\cdot k^-:S
		\quad
		\Gamma \bvdash Q\Subst{k^+}{k} \oftype \Delta\cdot k^+:S'
	\justifies
		\Gamma \bvdash P\Subst{k^-}{k} \mid Q\Subst{k^+}{k} \oftype
			\Delta\cdot \Delta' \cdot \; k^-:S \;\cdot \; k^+:S'
	\endprooftree
\justifies
	\Gamma \bvdash (\nu k)(\namedorch{k}{\orchF} \mid P\Subst{k^-}{k} \mid Q\Subst{k^+}{k} ) 
		\oftype \Delta\cdot \Delta'
\endprooftree
\]
}

\noindent
In case of rule [\textsc{OrchComm}] we have:
\[(\nu k)(\namedorch{k}{\io.\orchF} \mid \send{k^p}{e}.P \mid \receive{k^{\Dual{p}}}{x}.Q)
 ~\lts{}~  
(\nu k)(\namedorch{k}{\orchF} \mid P \mid Q\Subst{v}{x}).\]
Let us suppose without loss of generality that $p = -$ so that $\Dual{p} = +$. 
By hypothesis we have a derivation ending by:
{\small
\[
\prooftree
	\prooftree
		\Gamma \bvdash \send{k^-}{e}.P \oftype \Delta'' \cdot \; k^-:S
		\quad
		\Gamma \bvdash \receive{k^{+}}{x}.Q \oftype \Delta''' \cdot\;  k^+:S'
	\justifies
		\Gamma \bvdash \send{k^-}{e}.P \mid \receive{k^{+}}{x}.Q \oftype 
			\Delta \cdot \; k^-:S\; \cdot \; k^+:S'
	\endprooftree
\justifies
	\Gamma \bvdash (\nu k)(\namedorch{k}{\io.\orchF} \mid \send{k^p}{e}.P \mid \receive{k^{\Dual{p}}}{x}.Q)
		\oftype \Delta
\endprooftree
 \]
}
for some $S,S'$ such that $\orchComply{\io.\orchF}{S}{S'}$, and $\Delta'', \Delta'''$ such that $\Delta = \Delta'' \cdot \Delta'''$.
From the derivability of $\Gamma \bvdash \send{k^-}{e}.P \oftype \Delta'' \cdot \; k^-:S$ we deduce
that $\Gamma \vdash e:G$ for some ground $G$, 
$\Gamma \bvdash P \oftype \Delta'' \cdot \; k^-:S''$ and that $S = \;\outputtype{G}.S''$ for some $S''$.
Similarly we know that $\Gamma \vdash x:G'$ for some $G'$, 
$\Gamma, x:G' \bvdash Q \oftype \Delta''' \cdot\;  k^+:S'''$ and that $S' = \; \inputtype{G'}.S'''$.

Now from  this and the fact that $\orchComply{\io.\orchF}{\outputtype{G}.S''}{\inputtype{G'}.S'''}$ we infer
by Definition \ref{def:orchCompliance} that $G$ and $G'$ must be the same and that $\orchComply{\orchF}{S''}{S'''}$. Hence we have the derivation:
{\small
\[
\prooftree
	\prooftree
		\Gamma \bvdash P \oftype \Delta'' \cdot \; k^-:S'' 
		\quad
		\Gamma \bvdash Q\Subst{v}{x} \oftype \Delta''' \cdot \;k^+:S'''
	\justifies
		\Gamma \bvdash P \mid Q\Subst{v}{x} \oftype \Delta \cdot \;k^-:S'' \;\cdot\;k^+:S'''
	\endprooftree
\justifies
	\Gamma \bvdash (\nu k)(\namedorch{k}{\orchF} \mid P \mid Q\Subst{v}{x}) \oftype \Delta
\endprooftree
\]
}
where $\Gamma \bvdash Q\Subst{v}{x} \oftype \Delta''' \cdot \;k^+:S'''$ follows by $\Gamma \vdash e:G$, derivability of
$\Gamma, x:G \bvdash Q \oftype \Delta''' \cdot\;  k^+:S'''$, that $e\downarrow v$ and  Lemma \ref{lem:subst}.

\noindent In case of [\textsc{OrchDeleg}] we have:
\[(\nu k)(\namedorch{k}{\io.\orchF} \mid \throw{k^p}{{k'}^q}{P} \mid \catch{k^{\Dual{p}}}{{k'}}{Q})
 ~\lts{}~ 
(\nu k)(\namedorch{k}{\orchF} \mid P \mid Q\Subst{k'^q}{k'})\]
Then by hypothesis there exists the derivation:
{\small
\[
\prooftree
	\prooftree
		\prooftree
			\Gamma \bvdash P \oftype \Delta'  \cdot \; k^p:S_1
		\justifies 
			\Gamma \bvdash \throw{k^p}{{k'}^q}{P} \oftype \Delta'  \cdot \;  k'^q: S' \cdot\; k^p: \outputtype{S'^q}S_1
		\endprooftree
		\quad
		\prooftree
			\Gamma \bvdash Q\Subst{k'^r}{k'} \oftype \Delta''  \cdot \;  k'^r: S'' \cdot\; k^{\Dual{p}}:S_2
		\justifies
			\Gamma \bvdash \catch{k^{\Dual{p}}}{{k'}}{Q} \oftype \Delta''  \cdot \;  k^{\Dual{p}}: \inputtype{S''^r}S_2
		\endprooftree
	\justifies
		\Gamma \bvdash \throw{k^p}{{k'}^q}{P} \mid \catch{k^{\Dual{p}}}{{k'}}{Q} 
			\oftype \Delta \cdot \;  k'^q: S' \cdot\; k^p: \outputtype{S'^q}S_1  \cdot\; k^{\Dual{p}}: \inputtype{S''^r}S_2
	\endprooftree
\justifies
	\Gamma \bvdash (\nu k)(\namedorch{k}{\io.\orchF} \mid \throw{k^p}{{k'}^q}{P} \mid \catch{k^{\Dual{p}}}{{k'}}{Q}) 
		\oftype \Delta \cdot \; k'^q: S'
\endprooftree
\]
}
where 
$\Delta' \cdot \Delta'' = \Delta$; further, assuming w.l.o.g. that  $p = +$ and hence that $\Dual{p} = -$, we have
$\orchComply{\io.\orchF}{\outputtype{S'^q}S_1}{ \inputtype{S''^r}S_2}$, which implies that $S'' = S'$, $r = q$, and
$\orchComply{\orchF}{S_1}{S_2}$. Therefore there exists the derivation:
{\small
\[
\prooftree
	\prooftree
		\Gamma \bvdash P \oftype \Delta'  \cdot \; k^p:S_1
		\quad
		\Gamma \bvdash Q\Subst{k'^q}{k'} \oftype \Delta''  \cdot \;  k'^q: S' \cdot\; k^{\Dual{p}}:S_2
	\justifies
		\Gamma \bvdash P \mid Q\Subst{k'^q}{k'} \oftype \Delta \cdot \;  k'^q: S' \cdot\; k^p: S_1  \cdot\; k^{\Dual{p}}: S_2
	\endprooftree
\justifies
	\Gamma \bvdash (\nu k)(\namedorch{k}{\orchF} \mid P \mid Q\Subst{k'^q}{k'}) \oftype \Delta \cdot \; k'^q: S'
\endprooftree
\]
}
where the fact that $\Delta \cdot \;  k'^q: S' \cdot\; k^p: S_1  \cdot\; k^{\Dual{p}}: S_2$ is well defined follows by the fact
that $\Delta \cdot \;  k'^q: S' \cdot\; k^p: \outputtype{S'^q}S_1  \cdot\; k^{\Dual{p}}: \inputtype{S''^r}S_2$ is such.

\noindent
Let us consider the case
[\textsc{OrchSSel}]:
\[
(\nu k)(\namedorch{k}{\orchOplus_{h\in H}l_h.\orchF_h} \mid
\retrselop{k^p}{l_j:P_j}_{j\in J} \mid \branch{k^{\Dual{p}}}{\Set{l_i{:}Q_i}_{i\in I})} 
 ~\lts{}~
(\nu k)(\namedorch{k}{\orchF_c} \mid P_c \mid Q_c)
\]
where $c \in H\cup I\cup J$. By hypothesis we have a derivation ending by:
{\small
\[
\prooftree
	\Gamma \bvdash 
			\retrselop{k^p}{l_j:P_j}_{j\in J} \mid \branch{k^{\Dual{p}}}{\Set{l_i{:}Q_i}_{i\in I}} 
				\oftype \Delta \cdot \; k^p:\retrseltype{l_{j}{:}S_j}_{j\in J} \;\cdot\; k^{\Dual{p}}:\branchtype{l_{i}{:}S'_i}_{i\in I}
\justifies
	\Gamma \bvdash (\nu k)(\namedorch{k}{\orchOplus_{h\in H}l_h.\orchF_h} \mid
			\retrselop{k^p}{l_j:P_j}_{j\in J} \mid \branch{k^{\Dual{p}}}{\Set{l_i{:}Q_i}_{i\in I})} 
				\oftype \Delta
\endprooftree
\]
}
with the side condition $\orchComply{\orchOplus_{h\in H}l_h.\orchF_h}{\retrseltype{l_{j}{:}S_j}_{j\in J}}{\branchtype{l_{i}{:}S'_i}_{i\in I}}$. 
By Definition \ref{def:orchCompliance} this implies that $H \subseteq J \cap I$ is a non empty set such that for all $h\in H$ it holds
$\orchComply{\orchF_h}{S_h }{S'_h}$. The premise of the last inference in the above derivation must be derived by rule [\textsc{Conc-T}]
from
{\small 
\[
\begin{array}{c@{\hspace{6mm}}c@{\hspace{6mm}}c}
\prooftree
	\forall j \in J \quad \Gamma \bvdash P_j \oftype \Delta' \cdot \; k^p:S_j
\justifies
	\Gamma \bvdash \retrselop{k^p}{l_j:P_j}_{j\in J} \oftype 
		\Delta' \cdot \; k^p:\retrseltype{l_{j}{:}S_j}_{j\in J}
\endprooftree
&
\text{and}
&
\prooftree
	\forall i \in I \quad \Gamma \bvdash Q_i \oftype \Delta'' \cdot \; k^{\Dual{p}}:S'_i
\justifies
	\Gamma \bvdash \branch{k^{\Dual{p}}}{\Set{l_i{:}Q_i}_{i\in I}} \oftype \Delta'' \cdot \; k^{\Dual{p}}: \branchtype{l_{i}{:}S'_i}_{i\in I}
\endprooftree
\end{array}
\]
}
where $\Delta = \Delta'\cdot\Delta''$ is defined. Now since $c\in H \subseteq J\cap I$, we know that $\orchComply{\orchF_c}{ S_c }{ S'_c}$ so that
from the above we obtain the derivation:
{\small 
\[
\prooftree
\prooftree
		\Gamma \bvdash P_c \oftype \Delta' \cdot \; k^p:S_c
		\quad
		\Gamma \bvdash Q_c \oftype \Delta'' \cdot \; k^{\Dual{p}}:S'_c
	\justifies
		\Gamma \bvdash  P_c \mid Q_c \oftype \Delta \cdot \; k^p:S_c \;\cdot\; k^{\Dual{p}}:S'_c
	\endprooftree
\justifies
	\Gamma \bvdash (\nu k)(\namedorch{k}{\orchF_c} \mid P_c \mid Q_c) \oftype \Delta
\endprooftree
\]
}
The case of rule [\textsc{OrchSel}] is similar and simpler. Finally rules [\textsc{Par}], [\textsc{Scop}] and [\textsc{Str}]
follow by induction and Lemma \ref{lem:subjcongr}.
\end{proof}

\begin{corollary} [Error freeness] \label{cor:errorFree}
If $P$ is a user-defined process such that $\Gamma \bvdash P \oftype \Delta$ for some
$\Gamma$ and $\Delta$,
and $R$ is a run-time process such that $P \lts{*} R$, then $R$ is not an error.
\end{corollary}

Notice that, whereas our type system prevents deadlocks like the last one in 
Example \ref{ex:errorexamples}, it cannot prevent deadlocks due to the presence of subterms
like $(\nu k)(\send{k^+}{e}.R \mid R' \mid \namedorch{k}{\stopf})$, where $R'$ is not a $k^-$- process (see Example \ref{ex:complydependentdeadlock} and Definition \ref{def:compldeperr} below). This sort of deadllocks are intrinsecally due
to the asymmetric nature of our compliance relation. They can be avoided either by using
a less general compliance relation (namely a symmetric restriction of the present one), or
by extending  the operational semantics, as we shall do in the next section.


\section{$\comply$\,-dependent deadlocks and clean-up reductions}
\label{sec:cleanup}

Since in our setting session compatibility, unlike e.g duality, is asymmetric (compliance is a ''client biased'' relation), we have to face additional problems
related to peculiar deadlock states that we dub $\comply$\,-dependent deadlocks. 
These are due to actions on the server side that can remain unmatched after the client has reached 
a successful state, as shown in the following example. 

\begin{example}
\label{ex:complydependentdeadlock}
{\em 
\label{ex:vacorchdeadlock}
Let us consider the following processes:
$$
\begin{array}{l@{\hspace{8mm}}l}
P\ =\ \typedrequest{a}{\mathsf{S}_a}{k}{
                            \send{k}{4}.\typedrequest{b}{\mathsf{S}_b}{k'}{ \send{k'}{\text{True}}  } 
                                                                       }
&

Q\ =\ \typedaccept{c}{\mathsf{S}_c}{k}{
                            \receive{k}{x}.\send{k}{x}.\typedaccept{d}{\mathsf{S}_d}{k'}{ \receive{k'}{y}  } }

\end{array}
$$

where \hspace{4mm}
$
\mathsf{S}_a =\   \outputtype{\mathtt{Nat}}\hspace{6mm}
\mathsf{S}_b =\ \outputtype{\mathtt{Bool}}\hspace{6mm}
\mathsf{S}_c =\ \inputtype{\mathtt{Nat}}.\outputtype{\mathtt{Nat}}\hspace{6mm}
\mathsf{S}_d =\   \inputtype{\mathtt{Bool}}
$

\noindent
It is not difficult to see that the system $(P\, \mid\, Q)$
does type check and that its
evolution proceeds as follows:
The first session can be opened, since 
$\io.\stopf: \mathsf{S}_a \comply \mathsf{S}_c$, and we get\\
\centerline{
$(\nu k)(~\namedorch{k}{\io.\stopf} ~\mid~ \send{k^-}{4}.\typedrequest{b}{\mathsf{S}_b}{k'}{ \send{k'}{\text{True}} } ~\mid~ \receive{k^+}{x}.\send{k^+}{x}.\typedaccept{d}{\mathsf{S}_d}{k'}{ \receive{k'}{y}  }~)$
}
Now the orchestrator enables the input/output interaction on channel $k$, letting the system evolve to\\
\centerline{
$(\nu k)(~\namedorch{k}{\stopf} ~\mid~ \typedrequest{b}{\mathsf{S}_b}{k'}{ \send{k'}{\text{True}} } ~\mid~ \send{k^+}{4}.\typedaccept{d}{\mathsf{S}_d}{k'}{ \receive{k'}{y}  }~)$
}
The behaviour described by $\mathsf{S}_a$ has now been completed, but the $\send{k^+}{4}$ operation on the server side now prevents the system to progress.
}
\end{example}

We call $\comply$\,-dependent deadlock a stuck state like the one described in the above example. 

\begin{definition}
\label{def:compldeperr}
A process $P$ is a $\comply$\,-dependent deadlock if $P\notlts{}$ and it contains a subterm $R$
of the form\\
\centerline{
$R = (\nu k)(\namedorch{k}{\stopf} \mid R' \mid R'')$ ~~~~where $R'$ is a  $k^+$-process.}
\end{definition}

In order to avoid the above sort of deadlock states without modifying our notion of compliance, we extend the operational
semantics with some extra reductions for orchestrators like $\namedorch{k}{\stopf}$, i.e. the orchestrator that 
succesfully completed all the actions requested by the client on channel $k$.
In particular
the extra reductions formalized in Figure \ref{fig:cleanupred} allow an orchestrator $\namedorch{k}{\stopf}$ to vacuously satisfy all the synchronization actions on channel $k$ coming from the server,
namely those having $k^+$ in their prefix. We call them {\bf clean-up reductions}

\begin{figure}
\hrule
\[\begin{array}{l}
\mathax{\textsc{OrchClnUp1}}{
(\nu k)(\namedorch{k}{\stopf} ~\mid~ \pi.R' ~\mid~ R)~\lts{}~ (\nu k)(\namedorch{k}{\stopf} ~\mid~ R' ~\mid~ R)  ~~~~~~\textrm{where~} \pi {\in} \Set{\send{k^+}{e}, \receive{k^+}{x}, \select{k^+}{l}} 
}
\\[2mm]

\mathax{\textsc{OrchClnUp2}}{
(\nu k)(\namedorch{k}{\stopf} ~\mid~ \branch{k^+}{\Set{l_i{:}R'_i}_{i\in I}} ~\mid~ R)~\lts{}~ (\nu k)(\namedorch{k}{\stopf} ~\mid~ R'_c ~\mid~ R)
~~~~~~\textrm{if~} c\in I
}
\\[2mm]

\mathax{\textsc{OrchClnUp3}}{
(\nu k)(\namedorch{k}{\stopf} ~\mid~  \prtysel{k^+}{[l_i{:}Q_i]_{i\in I}} ~\mid~ R')~\lts{}~ (\nu k)(\namedorch{k}{\stopf} ~\mid~ R'_c ~\mid~ R) 
~~~~~~\textrm{if~} c\in I
}
\\ 
\end{array}
\]
\hrule
\caption{Clean-up reductions}\label{fig:cleanupred}
\end{figure}

By adding the clean-up  reductions to the operational semantics it is not difficult to check that the subject reduction property (Theorem \ref{thm:subjectReduction}) still holds. Out of that 
we get the following extension of Corollary \ref{cor:errorFree}.
\begin{corollary} [Error freeness] \label{cor:errorFreeExt}
If $P$ is a user-defined process such that $\Gamma \bvdash P \oftype \Delta$ for some
$\Gamma$ and $\Delta$,
and $R$ is a run-time process such that $P \lts{*} R$ (where also clean-up reductions are considered) then $R$ is neither an error nor a $\comply$\,-dependent deadlock.
\end{corollary}


Let us see how the $\comply$\,-dependent deadlock is avoided in Example \ref{ex:vacorchdeadlock}.

\begin{example}
{\em
Once we get to 
$$
(\nu k)(~\namedorch{k}{\stopf} ~\mid~ \typedrequest{b}{\mathsf{S}_b}{k'}{ \send{k'}{\text{True}} } ~\mid~ \send{k^+}{4}.\typedaccept{d}{\mathsf{S}_d}{k'}{ \receive{k'}{y}  }~)
$$
the reduction {\small \textsc{OrchClnUp1}} allows $\namedorch{k}{\stopf}$ to vacuously satisfy  the action$ \send{k^+}{4}$.
So, by such a rule,  the system can evolve to\\
\centerline{$
(\nu k)(~\namedorch{k}{\stopf} ~\mid~ \typedrequest{b}{\mathsf{S}_b}{k'}{ \send{k'}{\text{True}} } ~\mid~ \typedaccept{d}{\mathsf{S}_d}{k'}{ \receive{k'}{y}  }~)
$}
and then, since $\io.\stopf : \mathsf{S}_b \comply \mathsf{S}_d$, it can progress to
$
(\nu k)(~\namedorch{k}{\stopf} ~\mid~ (\nu k')(~ \namedorch{k'}{\io.\stopf}~\mid~
 \send{k'}{\text{True}}  ~\mid~ \receive{k'}{y}  ~) ~)
$\\
and finally to 
$
(\nu k)(~\namedorch{k}{\stopf} ~\mid~ (\nu k')(~ \namedorch{k'}{\stopf}~\mid~
 \inact  ~\mid~ \inact  ~) ~).
$
}
\end{example}

\begin{remark}
As previously mentioned, the necessity of clean-up reductions is due to the asymmetric nature of  the compliance relation. $\comply$\,-dependent deadlocks would not appear in typed terms if we forced compliance to be symmetric, for instance by replacing item \ref{def:orchCompliance-i} of 
Definition \ref{def:orchCompliance} with $\stopf : \stopA \comply \stopA$.
\end{remark}

%
\section{Deterministic-orchestrators and priority choices}
\label{sec:priorities}

As recalled in the introduction and discussed in e.g. \cite{Padovani10,BvBdL17},
the orchestration process should not 
exhibit any internal nondeterminism. In the present section, we stand by such a 
viewpoint and consider only deterministic orchestrators.
\begin{definition}[Deterministic orchestrators]\hfill\\
An orchestrator $f\in\Orch$ is {\em deterministic} if it does not contain any occurrence of the
$\oplus$ operator.
\end{definition}
Deterministic orchestrators can be used to select exactly one safe option in a construct like
$\retrselop{k}{l_i{:}P_i}_{i\in I}$ according to a given priority ordering among the 
failure-amenable options $\Set{l_i}_{i\in I}$. 
The priority ordering can be explicitely specified in the speculative types, which we 
interpret now as {\em speculative types with priorities}.
\begin{definition}[Types with priorities]
We modifiy the set $S$ of session types in Def. \ref{def:types} by replacing 
$\retrseltype{l_i{:}S_i}_{i\in I}$~ by ~~$\prtyretrseltype{l_1{:}S_1,\ldots,l_n{:}S_n}$, with $n\geq 1$.\\
The option represented by the label $l_i$ ($1\leq i\leq n-1$) is assumed to have higher priority than
the one represented by $l_{i+1}$.
\end{definition}

The process calculus remains unchanged, as well as the operational semantics, but for 
rule {\small [\textsc{Link}]}, which we now replace by 
\[\begin{array}{l}
 \mathax{\textsc{LinkPT}}{\typedrequest{a}{S}{k}{P} ~\mid~ \typedaccept{b}{S'}{k}{Q}  ~\lts{}~  
 (\nu k)(\namedorch{k}{\mathsf{f}} ~\mid~ P\Subst{k^-}{k} ~\mid~ Q\Subst{k^+}{k} )
\hspace{4mm}\textrm{if~~} \mathsf{f}=\Synth(S,S')\neq\FAIL  }
\end{array}
\]
where $\Synth$ is the orchestrator synthesis algorithm described in Figure \ref{fig:algSynth} and for which the following holds.
\begin{lemma}
\begin{enumerate}[i)]
\item
$\mathbf{Synth}(S,S')\neq\mathbf{fail}$\hspace{2mm} iff\hspace{2mm} $S\comply S'$;
\item
$\mathsf{f}=\mathbf{Synth}(S,S')\neq\mathbf{fail}$\hspace{2mm} implies \hspace{2mm}$\mathsf{f}:S\comply S'$ and $\mathsf{f}$ is deterministic.
\end{enumerate}
\end{lemma}

The proof of the above lemma is easy by definition of $\Synth$, which in turn can be easily obtained from that of orchestrated compliance.
Moreover, by construction, it is possible to show that the synthesized orchestrator always chooses the highest-priority option for {\em speculative types with priorities}.

\vspace{-2mm}
\paragraph{Process-level-specified priorities}
As mentioned in the introduction, it is possible to specify the priorities among failure-amenable
options at the process level instead of at the type level.
In order to do that, there is no need to change the set of types as provided in Definition \ref{def:types}.
We change instead the speculative selection operator.

\begin{definition}[Processes with priorities]
We modifiy the set $P$ of Processes in Def. \ref{def:processes} by replacing 
$\retrselop{k}{l_i{:}P_i}_{i\in I}$~ by ~~$\prtyretrselop{k}{l_1{:}P_1,\ldots,l_n{:}P_n}$, with $n\geq 1$.\\
The option represented by the label $l_i$ ($1\leq i\leq n-1$) is  assumed to have higher priority than
the one represented by $l_{i+1}$.
\end{definition}

The  operational semantics remains unchanged, but for the
rules {\small [\textsc{Link}]} and  {\small [\textsc{OrchSSel}]} which are now replaced by 
\[\begin{array}{l}
 \mathax{\textsc{LinkPP}}{\typedrequest{a}{S}{k}{P} ~\mid~ \typedaccept{b}{S'}{k}{Q}  ~\lts{}~  
 (\nu k)(\namedorch{k}{\mathsf{f}} ~\mid~ P\Subst{k^-}{k} ~\mid~ Q\Subst{k^+}{k} )
}
\\
\multicolumn{1}{r}{\textrm{if~} \mathsf{f}=\SynthUD(S,S')\neq\FAIL}\\

 \mathax{\textsc{OrchSSelPP}}{
(\nu k)(\namedorch{k}{\text{\small $\bigoplus$}_{i\in I}\mathsf{f}_i} ~\mid~\prtyretrselop{k}{l_1{:}P_1,\ldots,l_n{:}P_n} ~\mid~ \branch{k}{\Set{l_j{:}Q_j}_{j\in J}})
 ~\lts{}~
(\nu k)(\namedorch{k}{\mathsf{f}_c} ~\mid~ P_m ~\mid~ Q_c)
}\\
\multicolumn{1}{r}{\textrm{if~} c {\in} I\cap J,\, 1\leq m\leq n \textrm{~and~} l_c=l_m, \textrm{~where~}
m=\text{min}\Set{h\mid l_m\in\Set{l_i}_{i\in I} }   }
\end{array}
\]
where $\SynthUD$ returns the (possibly) non deterministic orchestrator exhibiting all 
safe options for speculative selection choices.
It can be easily obtainined out of $\Synth$ by replacing  the
third $\ELSE$ clause in Figure \ref{fig:algSynth} by the clause
{\small
\begin{tabbing}
 \ELSE ~~ \IF\ \= ($S = \branchtype{l_i{:}S_i}_{i\in I}$ \AND\ $S' = \prtyretrseltype{l_j{:}S'_j}_{j\in J} $) \OR\  ($S =\prtyretrseltype{l_j{:}S_j}_{j\in J} $ \AND\ $S' =  \branchtype{l_i{:}S'_i}_{i\in I}$) \ \\ 
	 \>\THEN\ \LET\ \=  $\text{res} = \SynthA\,(\List{S_h\comply S'_h}_{h\in I\cap J})$\ \\ 
	\> \>\IN\ \=\IF\ $\text{res}\neq$\ \List{} \ \= \THEN\  $l_{c_1}\Actdot \mathsf{f}_1\orchOplus\ldots \orchOplus l_{c_n}\Actdot \mathsf{f}_n $ \WHERE\ $\text{res} = [(\mathsf{f}_1,c_1),\ldots(\mathsf{f}_n,c_n)]$\\
\> \>\>\>\ELSE\ \FAIL
\end{tabbing}
}
\noindent where $\SynthA$ is defined by
{\small
\begin{tabbing}
$\SynthA(\List{\,})$ = \List{}  \\ [1mm]

$\SynthA$\=$((S_c\comply S'_c)\cons\xs )$ = \\
\>\LET\ \=$\mathsf{f}=\SynthUD(S_c \comply S'_c)$ 
 \IN\hspace{2mm} (\IF\ $\mathsf{f}=\ $\FAIL\ \THEN\ $\SynthA(\xs )$ \ELSE\ $(\mathsf{f},c)\cons  \SynthA(\xs )$) 
\end{tabbing}
}

\begin{remark}
The algorithms {\em \Synth} and {\em $\SynthUD$} can be adapted to the case of recursive orchestrators and types
following the treatment of recursion in \cite{BvBdL17}.
\end{remark}

\begin{figure}[ht]
\label{fig:algSynth}
\hrule
\vspace{2mm}
{\small
\begin{tabbing}
\Synth$(S\comply S')$\hspace{1mm} =\hspace{1mm} \=
\IF\ $S = \stopA$ \THEN\ $\stopf$ 
	\\ [2mm]

\>\ELSE \; \= \IF\ \= ($S = \inputtype{G}.S_1
$ \AND\ $S' = \outputtype{G}.S'_1$) \OR\  ($S = \outputtype{G}.S_1
$ \AND\ $S' = \inputtype{G}.S'_1$)  \ \\  [1mm]
\>	\> \>\THEN\ \LET\ \=  $\mathsf{f} = \Synth\,(S_1\comply S'_1)$\ \\ 
\>	\>\> \>\IN\ \IF\ $\mathsf{f}\neq$\ \FAIL\ \THEN\  $\io \Actdot \mathsf{f} $ \ELSE\ \FAIL
	\\ [2mm]

\>\ELSE \> \IF\ \= ($S = \inputtype{S_1^p}.S_2
$ \AND\ $S' = \outputtype{{S}^p_1}.S'_2$) \OR\  ($S = \outputtype{{S}^p_1}.S_2
$ \AND\ $S' = \inputtype{{S}_1^p}.S'_2$)  \ \\[1mm]
\>	\> \>\THEN\ \LET\ \= $\mathsf{f}= \Synth(S_2\comply S'_2)$\\
\>\>\>\>\IN\ \IF\= \   $\mathsf{f}\neq\FAIL$\  \THEN\  $\io \Actdot \mathsf{f} $ \ELSE\ \FAIL\\[2mm]


\> \ELSE \> \IF\ \= ($S = \branchtype{l_i{:}S_i}_{i\in I}$ \AND\ $S' = \prtyretrseltype{l_j{:}S'_j}_{j\in J} $) \OR\  ($S =\prtyretrseltype{l_j{:}S_j}_{j\in J} $ \AND\ $S' =  \branchtype{l_i{:}S'_i}_{i\in I}$) \ \\  [1mm]
\>	\> \>\THEN\ \LET\ \=  $\text{res} = \SynthL\,(\List{S_h\comply S'_h}_{h\in I\cap J})$\\ 
\>	\>\> \>\IN\ \=\IF\ $(\text{res}=$\ \FAIL) \THEN\  \FAIL\ \ELSE\ $ l_c \Actdot \mathsf{f} $ \WHERE\ $(\mathsf{f},c) = \text{res}$
	\\ [2mm]

\>\ELSE \> \IF\ \= ($S = \selectiontype{l_i{:}S_i}_{i\in I}$ \AND\ $S' = \branchtype{l_j{:}S'_j}_{j\in J} $) \OR\ ($S = \branchtype{l_j{:}S_j}_{j\in J}$ \AND\ $S' = \selectiontype{l_i{:}S'_i}_{i\in I} $)
			  \\ 
\>	 \>\> \THEN\ \LET\ \=  $\forall i\in I. \mathsf{f}_i =  \Synth(S_i, S'_i)$
\  \\ 
\>\>\> \>\IN\ \IF\ $\forall i\in I. \mathsf{f}_i \neq\ $ \FAIL\ \THEN\ $ \orchPlus_{i\in I} l_i.\mathsf{f}_i  $ \ELSE\ \FAIL \\ [2mm]

\>\ELSE \> $\FAIL$\\[2mm]

\>\hspace{-4mm}\WHERE \hspace{2mm}
\>$\SynthL(\List{\,})$ = \FAIL  \\ [1mm]

\>\>$\SynthL((S_c\comply S'_c)\cons\xs )$ = \LET\ \=$\mathsf{f}=\Synth(S_c \comply S'_c)$  \IN\ (\IF\ $\mathsf{f}\neq\ $\FAIL\ \THEN\ $(\mathsf{f},c)$ \ELSE\ $\SynthL(\xs )$)   
\end{tabbing}
\hrule
\vspace{-2mm}
}\caption{The algorithm \Synth.}\label{fig:Synth}
\vspace{4mm}

\end{figure}

%
\section{Conclusions and future work} \label{sec:conclusions}
\hspace{-2mm}
We have defined a type system with binary session types for a calculus with orchestrated interactions. The condition for session opening is here more permissive than in usual calculi  with session types in that the types of the processes participating in a session have just to be {\em compliant} rather than dual to each other.
The relation of compliance stems from the theory of contracts \cite{LP07,LP08,CGP10,BH16}. 
In particular, our compliance relation has been inspired by the orchestrated compliances proposed in  \cite{Padovani10} and \cite{BvBdL17},
where possible stuck states can be avoided by means of orchestrating processes.
In the present paper the points where an orchestrator can affect a computation are made visible in the calculus
by introducing a novel operator that we dub {\em speculative selection}. This implies 
extending the syntax of usual session types, so resulting in a session-type counterpart of
the retractable session contracts of \cite{BDLdL16}, where backtracking is taken into account
instead of orchestration.
Typable processes are shown to be free from errors.
An interpretation of speculative selection as an actual language construct is then provided
in form of priority selection. 

Recursion has not been considered in the present paper, in order to focus on the
main ideas and differences with other session-types formalisms.
Adding recursion should not pose any technical difficulties and can be dealt with,
for types and orchestrators, along the lines of \cite{Padovani10} and \cite{BvBdL17}.
Recursion for processes can be treated as done, among others, in \cite{HondaVK98}.\\
Besides adding recursion, the present formalism could be extended with a notion of 
subtyping. A subtyping relation for our type system can be formalized following ideas
present in  \cite{BdL13}. It should hence be possible to define our orchestrated compliance
by composing its symmetric restriction with the subtype relation.
For what concerns the treatment of higher-order in the subtype relation, one could
take into account investigations like those carried on in \cite{BdL13,BH16b} for the session-contracts formalism.
 
An interesting extension of the present investigation is adapting 
speculative choices and orchestrated interactions in the setting of multiparty asyncronous sessions.
In case one considered the multiparty session types of \cite{HYC16},
it would be reasonable to associate an orchestrator to each channel
involved in a multiparty session. The orchestrator would hence act as an input filter for the
buffer associated to a channel. 
Of course the introduction of speculative selection in a multiparty scenario
would imply to reconsider the relationship between 
global types and local types. Actually no speculative-selection type can reasonably come out by projecting a global type. So processes (and their related types) using speculative selections can be
looked at as modules one can adapt, by means of orchestrators, in order to 
comply (precisely or at least\footnote{In such a case, clean-up reductions would turn out to
be useful also in the multiparty setting, where it is reasonable not to have any particular bias
towards one of the participants in a multiparty session.}) with the global interaction pattern represented by the global type.
From a different perspective a global type, if any, could be synthesized out of a number
of local types, similarly to what done in \cite{Lange13}. In our ``speculative'' setting, however, a global type should rather be considered as a global adaptor, a centralized orchestrator,
something similar to the {\em medium process} of \cite{Caires16} whose action would now be restricted to driving the speculative choices.

By taking into account the discussion of Remark \ref{rem:dualityuptosub}, it would be interesting 
to define a {\em complementary relation} as the composition of compliance and subtyping
and investigating whether in our setting it can be interpreted as a generalization of duality.

\paragraph{Acknownledgements}
We are very grateful to the referees for their careful reading and useful suggestions.
We also thank Mariangiola Dezani-Ciancaglini for her support.

%
\label{sect:bib}

\bibliography{session}


\end{document}